\documentclass[11pt]{article}

\usepackage{a4wide}
 
\advance\topmargin-1.5cm
\usepackage{amssymb,filecontents,amsmath,amsthm, url,tikz, ifthen, enumitem}
\usepackage[breaklinks]{hyperref}

\def\alphabet{\Sigma}

\def\calB{\mathcal{B}}

\def\Columns{{\mathbb Z}/ 2\size {\mathbb Z}}
\def\Rows{\{0,\ldots, \size\}}
\def\DualColumns{ \{x+\tfrac12 \mid x \in \Columns\}}
\def\DualRows{\{  y+\tfrac12 \mid y \in \{-1,\ldots,\size\}\}}
\def\peq{p^=}
\def\pneq{p^{\neq}}
\def\muall{\mu_{k,d}}
\def\mueven{\muall^0}
\def\muodd{\muall^1}
\def\configs#1{#1\rightarrow \{0,1\}}
\def\terminals{T_{k,d}}
\def\gadget{C_{k,d}}
\def\oddterminal{T^\mathrm{1}}
\def\eventerminal{T^\mathrm{0}}
\def\oddterminals{\oddterminal_{k,d}}
\def\eventerminals{\eventerminal_{k,d}}
\def\prob#1#2#3{\begin{description}\item[\it Name]#1\item[\it Instance]#2\item[\it Output]#3\end{description}}
\def\PlanarTwoSpin{\textsc{DegreeFourPlanarTwoSpin}}
\def\PlanarLogTwoSpin{\textsc{PlanarLogTwoSpin}}
\def\PlanarHardCore{\textsc{DegreeFourPlanarHardCore}}
\def\PlanarCubicIS{\textsc{PlanarCubicIS}}
\def\matching{\mathcal{M}}
\def\size{\nu}
\let\epsilon=\varepsilon
\def\hatbeta{\widehat\beta}
\def\hatgamma{\widehat\gamma}
\def\hatlambda{\widehat\lambda}
\def\hatpeq{\widehat{\peq}}
\def\hatpneq{\widehat{\pneq}}
\def\hatP{\widehat{P}}
\def\hatM{\widehat{M}}
\def\hatW{\widehat{W}}
\def\hatp{\widehat p}
\def\NP{\mathrm{NP}}
\def\BPP{\mathrm{BPP}}
\def\RP{\mathrm{RP}}
\def\hatZ{\widehat{Z}}
\def\bsigma{\boldsymbol{\sigma}}
\def\bsigmasize{\bsigma_\size}
\def\bsigmakd{\bsigma_{k,d}}
\def\bsigmasub#1{\bsigma_{#1}}
\def\sigmarv{\bsigmakd}
\def\bsigmatilde{\tilde\bsigma}
\def\sigmatilde{\tilde\sigma}
\def\bsigmahat{\hat\bsigma}

\def\Sbar{\overline S}
\def\Int{\mathop{\mathrm{Int}}}
\def\ex{\mathop{\mathbb{E}\null}}
\def\Ext{\mathop\mathrm{Ext}}
\def\calBhat{\widehat{\mathcal{B}}}

\newcommand\wt{\operatorname{wt}}
 
\newtheorem{theorem}{Theorem}
\newtheorem{lemma}[theorem]{Lemma}
\newtheorem{corollary}[theorem]{Corollary}

\newtheorem{proposition}[theorem]{Proposition}
\newtheorem{definition}[theorem]{Definition}
\newtheorem*{mainprop}{Proposition \ref{prop:gadget}}
\newtheorem*{mainthm}{Theorem \ref{thm:main}}
\newtheorem*{thmthree}{Theorem \ref{thm:twospin_log_pras}}

\usetikzlibrary{decorations.markings,patterns}
\tikzset{ 
->-/.style={very thick, decoration={
  markings,
  mark=at position 0.6 with {\arrow{>}}},postaction={decorate}},
 shade/.style={pattern color=gray, pattern=dots}
}

\title {Approximating the partition function\\ of planar 
two-state spin systems\thanks{The research leading to these results has received funding from 
EPSRC grant EP/I011935/1 and from
the European Research Council under the European Union's Seventh Framework Programme 
(FP7/2007-2013) ERC grant agreement no.\ 334828. The paper 
reflects only the authors' views and not the views of the ERC or the European Commission. 
The European Union is not liable for any use that may be made of the information contained 
therein.
}}

\author{Leslie Ann Goldberg\thanks{
Department of Computer Science, University of Oxford, Wolfson Building, Parks Road, Oxford, OX1~3QD, UK.  
} \and Mark Jerrum\thanks{
School of Mathematical Sciences, Queen Mary, University of
London, Mile End Road, London, E1~4NS, UK. } 
\and Colin McQuillan\thanks{
Department of Computer Science, University of Liverpool,
Liverpool, L69~3BX, UK.}}

\begin{document}
\maketitle

\begin{abstract}

We consider the problem of approximating the partition function of the
hard-core model on planar graphs of degree at most~$4$.
We show that when the activity~$\lambda$ is sufficiently large,
there is no fully polynomial randomised approximation scheme
for evaluating the partition function unless $\NP=\RP$.
The result extends to a nearby region of the parameter space
in a more general two-state spin system with three parameters.
We also give a polynomial-time randomised approximation scheme
for the logarithm of the partition function.

\end{abstract}
 
 \section{Introduction}
    
A spin system is a model of particle interaction on a graph.   
Every vertex of the graph is assigned a state, called a spin.
A \emph{configuration} assigns a spin to every vertex,
and the weight of the configuration is determined by
interactions of neighbouring spins.
 
In this paper, we 
consider the following two-spin model,
which applies to spin systems on a graph $G=(V,E)$.  
The model has three parameters, $\beta$, $\gamma$ and $\lambda$. 
It is easiest to view these as non-negative rationals for now --- we
will be slightly more general later.
A configuration~$\sigma\colon\configs{V(G)}$ is an assignment of the two spins ``$0$'' and ``$1$''
to the vertices in~$V$.
The configuration~$\sigma$ has a \emph{weight} $w_G(\sigma)$, which
depends upon~$\beta$, $\gamma$ and~$\lambda$. 
Let $b(\sigma)$ denote the number of edges~$(u,v)$ of~$G$
with $\sigma(u)=\sigma(v)=0$,
let $c(\sigma)$ be the number of edges~$(u,v)$ of~$G$
with $\sigma(u)=\sigma(v)=1$
and let $\ell(\sigma)$ be the number of vertices~$u$ of~$G$
with $\sigma(u)=1$.
Then $w_G(\sigma) = \beta^{b(\sigma)} \gamma^{c(\sigma)} \lambda^{\ell(\sigma)}$.
The \emph{partition function}
of the model is given by
$$Z_{\beta,\gamma,\lambda}(G) = \sum_{\sigma:  \configs{V(G)} } w_G(\sigma).$$
 
Two important special cases are  
\begin{itemize}
\item the case $\beta=1$, $\gamma=0$, which is the \emph{hard-core model}, and
\item the case $\beta=\gamma$, which is the \emph{Ising model}. 
\end{itemize}
 
The \emph{hard-core model} \cite{vdBS} is a model of a gas 
in which vertices are 
either occupied by a particle (in which case they have spin~$1$) or unoccupied 
(in which case they have spin~$0$). 
The particles cannot overlap and adjacent vertices are close together, hence~$\gamma=0$. 
The Ising model is a model of ferromagnetism.
In this paper we study the hard-core model and a region of
nearby two-state spin systems.
 
 \subsection{Previous work}

Evaluating $Z_{\beta,\gamma,\lambda}(G)$ is 
a trivial computational problem if $\beta \gamma=1$,
because the partition function factors. In other cases, the complexity of
evaluation has been
studied in detail.  When $\lambda=1$, the problem of computing the
partition function on planar $\Delta$-regular graphs is
called $\operatorname{Pl-Hol}_\Delta(a,b)$ in \cite{CK}, 
where $a$ corresponds to~$\beta$ and $b$ corresponds to~$\gamma$.  
Assume $\Delta\geq3$.  
There is a dichotomy 
\cite[Theorem 1]{CK}: for non-negative $a,b$, the problem
$\operatorname{Pl-Hol}_\Delta(a,b)$ can be computed exactly in polynomial time in the
trivial cases $a b=1$ and $a=b=0$, and in the case of the Ising model
with no external field, $a=b$.  In all other cases, the problem of exactly computing the
partition function is
\#P-hard.

A standard transformation extends this dichotomy to arbitrary~$\lambda>0$. 
Consider a configuration $\sigma\colon \configs{V(G)}$ of a planar $\Delta$-regular graph~$G$.
Counting the number of edges adjacent to a ``$1$'' spin in two ways,
we have
$\Delta \ell(\sigma)=2c(\sigma)+(|E(G)|-b(\sigma)-c(\sigma))$.
Therefore,
$$Z_{\beta,\gamma,\lambda}(G)=\lambda^{|E(G)|/\Delta}Z_{\beta\lambda^{-1/\Delta},\gamma\lambda^{1/\Delta},1}(G),$$  which is as hard to compute as
 $\operatorname{Pl-Hol}_\Delta(\beta\lambda^{-1/\Delta},\gamma\lambda^{1/\Delta})$. 
 Suppose $\beta$ and $\gamma$ are not both~$0$.
 Unless
 $\lambda = 1$, we have either
 $\beta\lambda^{-1/3}\neq\gamma\lambda^{1/3}$ or
 $\beta\lambda^{-1/4}\neq\gamma\lambda^{1/4}$.   
 If $\beta \gamma \neq 1$ 
 then
 in either case,
 we can
 conclude from above that
 evaluating $Z_{\beta,\gamma,\lambda}(G)$ is \#P-hard 
when the input~$G$ is restricted to be a planar graph of degree at most~$4$.
  
Since the complexity of exactly evaluating the partition function is intractable,  much effort has focussed on the difficulty of
\emph{approximately} evaluating the partition function for a given set of parameters~$\beta$, $\gamma$ and 
$\lambda$.  

The complexity of approximating the partition function of
the hard-core model and the Ising model 
in general (not necessarily planar) graphs
is well-understood.
The \emph{Gibbs measure} 
is the distribution on
configurations~$\sigma\colon\configs{V(G)}$ 
in which the probability of configuration~$\sigma$ is proportional to $w_G(\sigma)$.
This notion of Gibbs measure extends to certain infinite graphs, for example infinite 
regular trees, where it may or may not be unique. 
For the hard-core model, 
there is a critical point $\lambda_c(\Delta) = (\Delta-1)^{\Delta-1}/(\Delta-2)^\Delta$
such that the infinite $\Delta$-regular tree has a  
unique Gibbs measure if and only if $\lambda \leq \lambda_c$.
An important result of Weitz~\cite{Weitz} showed that, in 
\emph{every} graph with maximum degree at most~$\Delta$,
the correlations between spins in the hard-core model decay
rapidly with distance as long as $\lambda \leq \lambda_c$.
As a result, he gives~\cite[Corollary 2.8]{Weitz} a fully-polynomial (deterministic) approximation scheme
(FPTAS)
for evaluating the hard-core partition function  
on graphs of degree at most~$\Delta$
for any $\lambda < \lambda_c$.
By contrast,   Sly and Sun~\cite[Theorem 1]{SS12} 
(see also the earlier hardness results
of Sly~\cite{Sly} and Galanis et al.~\cite{GGSVY})
show
that, unless $\NP=\RP$, there is no 
fully-polynomial randomised approximation scheme
(FPRAS) on $\Delta$-regular graphs (for $\Delta\geq 3$)
for any $\lambda> \lambda_c(\Delta)$. Thus, the difficulty of
approximation is resolved, apart from at the boundary $\lambda=\lambda_c(\Delta)$.

We say that the two-spin model is \emph{ferromagnetic} if $\beta \gamma>1$ and
\emph{antiferromagnetic} if $\beta\gamma<1$.
For the antiferromagnetic Ising model, Sinclair et al.~\cite[Corollary 1]{Sinclair}
show that there is a FPTAS for evaluating the Ising partition function
on graphs of degree at most~$\Delta$
for any choice of parameters~$\beta$ and~$\lambda$ which is in the interior of the uniqueness
region of the $\Delta$-ary tree.
By contrast,  Sly and Sun~\cite[Theorem 2]{SS12} show that,
unless $\NP=\RP$, there is no FPRAS on $\Delta$-regular graphs (for $\Delta\geq 3$)
if $\beta$ and $\lambda$ are outside the uniqueness region. (So, once again, the situation is
fully resolved, apart from the boundary.)
The result of Sinclair et al.\ extends
to general anti-ferromagnetic two-state spin systems
in regular graphs, and also in a somewhat wider class of graphs~\cite[Corollary 2]{Sinclair}.

For general anti-ferromagnetic two-state spin systems, the best positive
result that is known is due to Li, Lu, and Yin~\cite{LLY12}. They use a stronger notion of
correlation decay than Weitz, which enables them to obtain  
a PTAS, even for graphs with unbounded degree.
They show~\cite[Theorem 2]{LLY12}
that for any finite $\Delta\geq 3$, or for $\Delta=\infty$, there is an FPTAS for
the partition function of the two-state spin system on graphs of
maximum degree at most~$\Delta$ if the parameters of
the system are antiferromagnetic, and for every $d\leq \Delta$, they lie in the interior of the uniqueness
region of the infinite $d$-regular tree.
By contrast~\cite[Theorem 3]{LLY12},
the results of Sly and Sun imply that, for any finite $\Delta\geq 3$,
or for $\Delta=\infty$, unless $\NP=\RP$, there is no FPRAS for
the partition function of the two-state spin system on graphs
of maximum degree at most~$\Delta$ if
the parameters of the system are antiferromagnetic, and
 for some $d\leq \Delta$,
they lie outside the interior of the uniqueness region of the infinite $d$-regular tree.
Thus, the approximation complexity is resolved in the antiferromagnetic case,
apart from at the boundaries of the uniqueness regions.
Note that the result of Sun and Sly was independently discovered by Galanis, 
\v Stefankovi\v c and Vigoda~\cite{GSV} for the case $\lambda=1$.

The situation is not completely resolved in the ferromagnetic case.
Building on Jerrum and Sinclair's FPRAS for the ferromagnetic Ising model~\cite{JS},
Goldberg, Jerrum and Paterson~\cite{GJP} gave an FPRAS for the ferromagnetic
two-spin model which applies if
$\beta \geq \gamma$ and $\lambda \leq \sqrt{\beta/\gamma}$
(or, equivalently, if
$\beta \leq \gamma$ and $\lambda \geq \sqrt{\beta/\gamma}$).
The approximation applies without these constraints on the parameters if
the input is a regular graph.

For the hard-core model, an important issue which arises
in statistical physics
is  approximating the partition function for planar graphs,
including regular lattices.
While (as far as we know) there were no hardness results
for this problem (until this paper)
the complexity of particular algorithms have been studied.
For example,
Randall~\cite{Randall} showed
that a particular MCMC algorithm provides a bad approximation
on subsets of~$\mathbb{Z}^2$,
because Glauber dynamics mixes slowly when $\lambda \geq 8.066$.
(By contrast, results of Restrepo et al.\ \cite{Restrepo} 
showed that the mixing time is $O(n \log n)$ 
when $\lambda < 2.3883$,
and that  Weitz's algorithm~\cite{Weitz} gives a (deterministic) fully-polynomial-approximation
scheme in this case.)
Recently tree decompositions of planar graphs have been used to give FPTASes for certain partition functions on planar graphs --- see~\cite{YZ12}.

\subsection{Our contribution}

Our objective is to determine whether 
approximating the
partition function of the hard-core model is computationally intractable 
on planar graphs for sufficiently large~$\lambda$.
It turns out that this is so.
Our main result (see Theorem~\ref{thm:main}) is that,
for a wide range of two-spin parameters, 
there is no FPRAS, even for planar graphs with degree at most~$4$.
The applicable range of parameters includes
the hard-core model with $\lambda \geq 312$.
Thus, we show that approximation is difficult for
this problem (see Corollary~\ref{cor:hardcore}).

An interesting difference between the general case and the planar
case is that, in general, it is difficult to approximate the logarithm of the
partition function, a quantity which has physical significance and is called the
\emph{mean free energy}.
Sly and Sun (see the proofs of Theorems~1 and~2 in~\cite{SS12}) showed that there is a fixed $c>1$ such
that no algorithm can approximate $Z_{\beta,\gamma,\lambda}(G)$
within a factor $c^{|V(G)|}$ unless $\NP=\RP$.
By contrast, we show (see Theorem~\ref{thm:twospin_log_pras})
that, in the planar case, there is a polynomial-time approximation scheme
for $\log Z_{\beta,\gamma,\lambda}(G)$.  Note that this result implies that, 
for any $c>1$, the partition function $Z_{\beta,\gamma,\lambda}(G)$
can be approximated within a factor $c^{|V(G)|}$ (since $Z_{\beta,\gamma,\lambda}(G)$
is at most $C^{|V(G)|}$ for a quantity~$C$ which depends only on~$\beta$, $\gamma$ and $\lambda$).
  
At a high level, our hardness result is
a reduction from the optimisation problem of computing a maximum independent set in 
a cubic planar graph~$G$ to the problem of estimating the partition function
of a much larger graph, which is constructed from~$G$.
Each vertex of~$G$ is represented by a 
gadget which is a ``wrapped'' rectangular lattice~$C_\size$
(see Figure~\ref{fig:one}). 
Similar to previous results of Goldberg and Jerrum~\cite{GJPottsconf}, and
Sly~\cite{Sly}, and Sly and Sun~\cite{SS12}, 
we exploit the phase transition of the gadget
to enable a reduction from a hard optimisation problem.
The optimisation problem from which we start
(computing a maximum independent set in a cubic planar graph) plays
a similar role to
that of the maximum cut problem in the reduction of Sly and Sun~\cite{SS12}.
However, there is a key difference.
Since, as we discuss below, 
it turns out that the logarithm of our partition function is
efficiently approximable,
it is therefore necessary that the optimisation problem 
from which we start  is also easy to approximate (otherwise, we would get
a contradiction). This means that our reduction has to be more carefully tuned ---
the approximation of the partition function has to allow us to exactly solve
the optimisation problem.

A key technical challenge in the proof is to characterise the Gibbs distribution of the
two-spin model on the lattice gadget.
We show that the spins of the vertices do exhibit  long-range correlation.  
In fact, the gadget is almost always
in one of two phases.
Each of these phases are equally likely. Also, conditioned on the
phase, the spins of certain vertices along the boundary of the gadget are
nearly independent, and 
their distribution
can be determined. 
Thus, although there is long-range correlation between spins, all of the
correlation is captured by the phase. Conditioned on the phase,
the spins are not very correlated.
The analysis of
the
Gibbs distribution of the gadget uses 
contour arguments adapted from Dobrushin~\cite{Dob68} and Borgs et al.~\cite{Borgs}.
Randall's slow-mixing result is also based on contour arguments. 

In statistical physics  it is sometimes  useful to
approximate the logarithm of the partition function,
even when the  partition function itself cannot be approximated (for example, in the situation of
Theorem~\ref{thm:main}).  
Bandyopadhyay and Gamarnik \cite{Gamarnik} have shown how to estimate the
logarithm of the partition function of the hard-core model when $\lambda$ is small
and the graph is regular, with large girth. They show that, in this case, the approximate value
does not depend on the graph, given its degree and size!
We give (Theorem~\ref{thm:twospin_log_pras})
an
approximation scheme for the logarithm of the partition function
which applies to all planar graphs, for sufficiently large~$\lambda$.
The algorithm is based on the decomposition technique that Baker~\cite{Baker} 
used to give approximation schemes for optimisation problems on planar graphs.
There is a parameter~$k$ which is governed by
desired approximation quality. The graph $G$ is decomposed into pieces 
which are $k$-outerplanar, and therefore have bounded tree-width.
The partition functions of these pieces can be calculated directly using an algorithm
of Yin and Zhang~\cite{YZ12}. These are combined to give the estimate.

\section{Preliminaries and statement of results} 
 
In our main result, we will assume that the parameters  $\beta$, $\gamma$ and $\lambda$ 
satisfy the following conditions.
\begin{equation}  \label{eq:spincond}
\lambda \geq 1,\quad
\beta \geq 1 > \gamma\geq 0,\quad
\beta \gamma<1,\quad
\beta \lambda^{-1/4} \leq 0.238, \mbox{ and }
\gamma \lambda^{3/8} \leq 0.238.
\end{equation}   
Note that these conditions are satisfied by the hard-core model
when $\lambda\geq 312$ (by setting $\beta=1$ and $\gamma=0$).

The notion of a fully polynomial randomised approximation scheme (FPRAS)
is defined in Section~\ref{sec:FPRAS}.   
Following~\cite{GJPotts},    
we  say that a real number~$z$ is \emph{efficiently approximable} if there is an FPRAS
for the problem of computing~$z$. 
For fixed efficiently approximable reals
$\beta$, $\gamma$ and $\lambda$ satisfying
 \eqref{eq:spincond}, we consider
 the problem of (approximately) computing $Z_{\beta,\gamma,\lambda}(G)$,
 given an input graph~$G$.
 In order to make our (negative) result as strong as possible,
 we restrict the input~$G$ to have degree at most~$4$
 as well as being planar. Thus, we study the 
 following computational problem.
 \prob{$\PlanarTwoSpin(\beta,\gamma,\lambda)$.}
 {A planar graph $G$ with maximum degree at most~$4$.}
 {The value $Z_{\beta,\gamma,\lambda}(G)$.}
     
Our main result is the following.
   
\begin{theorem}\label{thm:main} 
Suppose that $\beta$, $\gamma$ and $\lambda$ are efficiently approximable reals
satisfying \eqref{eq:spincond}.
There is no FPRAS for $\PlanarTwoSpin(\beta,\gamma,\lambda)$
unless $\NP=\RP$.
\end{theorem}

Of course, our result has an immediate consequence for the 
problem of approximating the partition function in the hard-core model.
Thus, Theorem~\ref{thm:main} implies Corollary~\ref{cor:hardcore}
for the following computational problem.
\prob{$\PlanarHardCore(\lambda)$.}
{A planar graph $G$ with maximum degree at most~$4$.}
{The  value $Z_{1,0,\lambda}(G)$.}
\begin{corollary} \label{cor:hardcore}
Suppose that
$\lambda\geq  312$ is an efficiently approximable real.
There is no FPRAS for  $\PlanarHardCore(\lambda)$ unless $\NP=\RP$.
\end{corollary}

Despite Theorem~\ref{thm:main},
we show that the logarithm of the partition function can be approximated.
In particular, we study the following computational problem, where, for concreteness,
we use the natural logarithm (to the base~$e$).
   
 \prob{$\PlanarLogTwoSpin(\beta,\gamma,\lambda)$.}
 {A planar graph $G$.}
 {The value $\log(Z_{\beta,\gamma,\lambda}(G))$.}
  
 Our result is that there is a polynomial-time randomised approximation scheme
 (PRAS) for $\PlanarLogTwoSpin(\beta,\gamma,\lambda)$. A polynomial-time
 randomised approximation scheme is a more liberal notion of approximation than 
 the notion of an FPRAS. See Section~\ref{sec:FPRAS} for a definition.  
   
 \begin{theorem}\label{thm:twospin_log_pras}
Suppose that $\beta$, $\gamma$ and $\lambda$ are efficiently
approximable reals satisfying  
$\beta\geq 1 > \gamma\geq 0$ and $\lambda\geq 1$.
There is a PRAS for $\PlanarLogTwoSpin(\beta,\gamma,\lambda)$.
\end{theorem} 
     
The randomness used by the
algorithm promised by Theorem~\ref{thm:twospin_log_pras}
is only needed  to approximate the parameters~$\beta$, $\gamma$ and~$\lambda$.
If these are deterministically approximable, then the approximation is deterministic.

We will need some notation 
to refer to the Gibbs distribution of the two-spin model
on a graph~$G$, which is the distribution 
in which the the probability of
each configuration is proportional to its weight.
We will use $\bsigma_G$ to denote
a random configuration drawn from this distribution.
Thus, for any configuration $\sigma\colon \configs{V(G)}$,
$$\Pr(\bsigma_{G} = \sigma) = w_{G}(\sigma)/Z_{\beta,\gamma,\lambda}(G).$$
(In general, as here, we use boldface for the random variable and normal 
type for the values that it takes on.)
Finally, given a subset $S$ of $V(G)$ and a configuration $\sigma: \configs{V(G)}$,
let $\sigma(S): \configs{S}$ denote the configuration induced by~$\sigma$ on~$S$.

\section{ Polynomial Randomised Approximation Schemes}  
\label{sec:FPRAS}

Most of this section is taken from~\cite{GJPotts}
and can be skipped by readers who are already familiar with  
randomised approximation schemes.

A randomised approximation scheme is an algorithm for
approximately computing the value of a function~$f:\alphabet^*\rightarrow
\mathbb{R}$.
(Here, $\alphabet$ is a finite alphabet, and inputs to~$f$ are represented as
strings over this alphabet.)
The
approximation scheme has a parameter~$\varepsilon>0$ which specifies
the error tolerance.
A \emph{randomised approximation scheme\/} for~$f$ is a
randomised algorithm that takes as input an instance $ x\in
\alphabet^{\ast }$ (e.g., for the problem $\PlanarTwoSpin(\beta,\allowbreak\gamma,\lambda)$, the
input would be an encoding of  a planar graph~$G$) and a rational error
tolerance $\varepsilon \in(0,1)$, and outputs a rational number $z$
(a random variable of the ``coin tosses'' made by the algorithm)
such that, for every instance~$x$,
\begin{equation}
\label{eq:3:FPRASerrorprob}
\Pr \big[e^{-\epsilon} f(x)\leq z \leq e^\epsilon f(x)\big]\geq \frac{3}{4}\, .
\end{equation}
The randomised approximation scheme is said to be a
\emph{polynomial
randomised approximation scheme} or \emph{PRAS}
if, for each $\epsilon$, its running time is bounded by a polynomial in~$|x|$.
It is said to be a
\emph{fully polynomial randomised approximation scheme},
or \emph{FPRAS},
if its running time is bounded by a polynomial
in $ |x| $ and $ \epsilon^{-1} $.

Note that the quantity $\frac34$ in
Equation~(\ref{eq:3:FPRASerrorprob})
could be changed to any value in the open
interval $(\frac12,1)$ without changing the set of problems
that have randomised approximation schemes \cite[Lemma~6.1]{jvv}.
In fact, in the proof of Theorem~\ref{thm:main}, we will assume that our FPRASes
have failure probability at most~$1/15$.

The notion of an FPRAS is a particularly
robust notion of approximability for partition functions.
For such approximations, the existence of a polynomial-time algorithm that achieves a constant-factor approximation actually implies the existence of an FPRAS. 
The same argument that
we gave to illustrate this point for the Potts model~\cite{GJPotts} also applies to the setting of this paper. For any graph $G$, denote by $k\cdot G$ the graph composed of $k$ disjoint copies of~$G$.
Then $Z_{\beta,\gamma,\lambda}(k \cdot G) = Z_{\beta,\gamma,\lambda}(G)^k$.  So, setting $k=O(\epsilon^{-1})$,
a constant factor approximation to $Z_{\beta,\gamma,\lambda}(k \cdot G)$ will yield (by taking the $k$th root)
an FPRAS for $\PlanarTwoSpin(\beta,\gamma,\lambda)$. Clearly, an approximation within a polynomial factor would 
also suffice.
Note that the same argument does not necessarily apply to log-partition functions.

\section{The Gadget}

We will assume throughout this section  
  that 
$\beta$, $\gamma$ and $\lambda$ satisfy \eqref{eq:spincond}, so we do not
keep repeating this condition in the statement of our lemmas.

The gadget~$C_\size$  has vertex set
$V(C_\size)=\Columns \times \Rows$. 
Vertices $(x,y)$ and $(x',y')$ are adjacent in~$C_\size$ if
\begin{itemize}
\item $y=y'$ and $x=x'\pm 1$ (where of course, the arithmetic is modulo~$2\size$ since~$x$ and~$x'$ are
in $\Columns$), 
or
\item $x=x'$ and $y=y' \pm 1$.
\end{itemize}
Let $E(C_\size)$ denote the set of edges of~$C_\size$.
See the leftmost picture in Figure~\ref{fig:one}. 
 \begin{figure}[!h]
\begin{center}
 \begin{tikzpicture}
\begin{scope}[xshift=0cm]
\foreach \x in {0,...,9} \draw
 (canvas polar cs:angle=-90+36*\x,radius=1cm) --
 (canvas polar cs:angle=-90+36*\x,radius=2cm);
\foreach \y in {0,...,5} \draw (0,0) circle (2cm-0.2*\y cm);
\foreach \x in {0,...,9} \foreach \y in {0,...,5}
   \fill (canvas polar cs:angle=-90+36*\x,radius=2cm-0.2*\y cm) circle (0.05cm);
\foreach \x in {-4,...,5} \draw
 (canvas polar cs:angle=-90+36*\x,radius=2.7cm) node {(\x,0)};
\end{scope}

\begin{scope}[xshift=6cm]
\foreach \x in {0,...,9} \draw
 (canvas polar cs:angle=-90+36*\x,radius=1cm) --
 (canvas polar cs:angle=-90+36*\x,radius=2cm);
\foreach \y in {0,...,5} \draw (0,0) circle (2cm-0.2*\y cm);
\foreach \x/\y in {-1/0,-1/1,-1/2,0/2,1/2,2/2,3/2,3/1,3/0}
   \fill (canvas polar cs:angle=-90+36*\x,radius=2cm-0.2*\y cm) circle (0.05cm);
\foreach \x in {1} \draw
 (canvas polar cs:angle=-90+36*\x,radius=2.7cm) node {(\x,0)};
\end{scope}

\begin{scope}[xshift=11cm]
\foreach \x in {0,...,9} \draw
 (canvas polar cs:angle=-90+36*\x,radius=1cm) --
 (canvas polar cs:angle=-90+36*\x,radius=2cm);
\foreach \y in {0,...,5} \draw (0,0) circle (2cm-0.2*\y cm);
\foreach \x in {-9,...,9}
   \fill (canvas polar cs:angle=-90+36*\x,radius=2cm-0.2*5 cm) circle (0.05cm);
\foreach \y in {0,...,5}
   \fill (canvas polar cs:angle=-90+36*5,radius=2cm-0.2*\y cm) circle (0.05cm);
\foreach \x in {0,5} \draw
 (canvas polar cs:angle=-90+36*\x,radius=2.7cm) node {(\x,0)};
\end{scope}

\end{tikzpicture}

\end{center}
\caption{$C_{5}$, and the vertex subsets $B_{1,2}$, and $B_{0,5}$.}
\label{fig:one}
\end{figure}
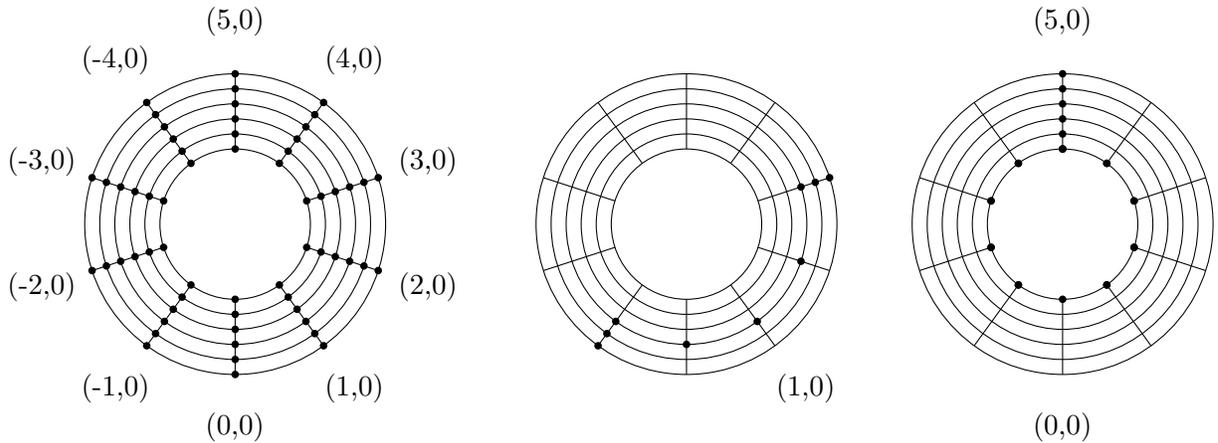

\subsection{Goalposts and keyholes}

Given a  vertex $(x,0)\in V(C_\size)$ and a value $m\in \Rows$
let $B_{x,m}$ be the set containing the vertices on the rectangular (goalpost-shaped) path
at $\ell_\infty$-distance~$m$ around the terminal. In particular, 
  let
$$ B_{x,m} =   \bigcup_{0\leq j \leq m} 
\{(x-m,j),(x-j,m),(x+j,m),(x+m,j)\}.$$
Again, the arithmetic is done modulo~$2 \size$ since $x\in \Columns$.
See the middle picture in Figure~\ref{fig:one}.

When $m= \size$, the vertices in $\{(x-m,j) \mid 0 \leq j \leq m\}$
coincide with the vertices in $\{(x+m,j) \mid 0 \leq j \leq m\}$ so $B_{x,m}$
becomes the ``keyhole'' which is depicted in 
the rightmost picture of Figure~\ref{fig:one} (for $x=0$).

We shall often be working with configurations on gadgets.
For convenience the notation $\bsigmasub{C_\size}$ will be contracted to $\bsigmasize$,
and no confusion should result.

\subsection{Parity-$0$ ones and parity-$1$ ones}

We say that a vertex $(x,y) \in V(C_\size)$ 
has parity~$1$
if $x+y$ is odd, and that it has parity~$0$ otherwise.
Suppose that $S$ is a subset of $V(C_\size)$ and that $s\in\{0,1\}$.
We say that $\sigma(S)$ has \emph{parity-$s$ ones}
if $\{ (x,y) \in S \mid \sigma(x,y)=1\}$ is exactly the set of  parity-$s$ vertices in~$S$.

 \subsection{Idealised probabilities}

 Define 
\begin{align*}
 \peq &= \limsup_{\size\rightarrow \infty}
\Pr(\bsigmasize(0,0)=1 \mid \mbox{$\bsigmasize(B_{0,\size})$ has  parity-$0$ ones}), \mbox{and}\\
\pneq &= \limsup_{\size\rightarrow \infty}
\Pr(\bsigmasize(1,0)=1 \mid \mbox{$\bsigmasize(B_{1,\size})$ has parity-$0$ ones}). 
\end{align*}

The notation $\peq$ is meant to connote that we are looking at the
probability of a $1$ at a vertex of parity $s$, conditioned on certain
parity-$t$ ones, where $s=t$; for $\pneq$ we are interested in $s\neq
t$.  As we shall see later, it will turn out that $\peq > \pneq$.
This is a non-trivial fact about the spin system: if there were no
long-range correlations, we would have $\peq=\pneq$.
The following straightforward lemma is also useful.
 
\begin{lemma}
\label{lem:little}
$\pneq>0$ and $\peq<1$. 
\end{lemma}
\begin{proof}
Suppose $\size\geq 2$.
Consider  vertex $(1,0)$ of  $C_\size$.  
Let $S = \{ (1,0),(2,0),(0,0),(1,1) \}$ be the set containing $(1,0)$ and its
immediate neighbours.
Let $S' = \{ (-1,0),(0,1),(1,2),(2,1),(3,0) \}$ be the set containing the neighbours of~$S$.
Given any $\sigma:\configs{S'}$,
$$\Pr(\mbox{$\bsigmasize(S)$  has  parity-$1$ ones} \mid
\bsigmasize(S')=\sigma) \geq \lambda / 16 \lambda^4 \beta^{10}
> 0
.$$  
 
 Now let $S''=\{(-1,0),(0,1),(1,0)\}$ be the neighbours of $(0,0)$.
Given any $\sigma:\configs{S''}$,
$$\Pr(\mbox{$\bsigmasize(0,0)=1$} \mid
\bsigmasize(S'')=\sigma) \leq \frac{\lambda}{1+\lambda}<1.$$
\end{proof} 

The events that $\bsigmasize(B_{0,\size})$ has  parity-$0$ ones 
and that $\bsigmasize(B_{1,\size})$ has  parity-$0$ ones 
have low probability, so it
may seem strange to condition on these events, but the purpose of this conditioning is
to identify two  phases of the idealised gadget.
We will refer to certain vertices $(x,0)$ of $C_\size$
as ``terminals'', and  it will turn out to be the case that  the spins of
these terminals are nearly independent of each other in the distribution of~$\bsigmasize$.

We will study the distribution that $\bsigmasize$ induces on the terminals
by considering an idealised distribution with two  phases.
In each of these two gadget  phases, the spins of the terminals will be chosen independently.
Some terminals will be assigned spin~$1$ with probability~$\peq$ and others
will be assigned spin~$1$ with probability~$\pneq$. This will be explained further in the next section.

\subsection{Terminals}
\label{sec:terminals}

Fix positive integers~$d$ and $k$.
Let $\size= 2 d k$ and let $\gadget$ denote the gadget $C_\size$.
We will work with $\gadget$ for the rest of the paper.
We will use both notations, $C_\size$ and $\gadget$, depending on whether we want
to emphasize the role of~$\size$ or the role of~$k$ and~$d$.
Similarly, the alternative notations, $\sigmarv$ and $\bsigmasize$ will be used as convenient.

Some of the vertices around the boundary of $\gadget$
($2k$~of them) are designated as ``terminals''.
The set of ``parity-$1$ terminals'' is
$$\oddterminals = \{ (4 j d + 1,0) \mid 0 \leq j \leq k-1\}.$$
The set of ``parity-$0$ terminals'' is
$$\eventerminals = \{ (4 j d + 2 d,0) \mid 0 \leq j \leq k-1\}.$$ 
Let $\terminals = \oddterminals \cup \eventerminals$ denote the set of terminals.

For parity~$s\in \{0,1\}$,
let $\muall^s$ 
be the distribution on configurations $\sigma:\configs{\terminals}$
in which the spin of each terminal is chosen independently as follows:
For each parity-$s$ terminal $(x,0)$, 
set $\sigma(x,0)=1$
with probability~$\peq$ (and set $\sigma(x,0)=0$ otherwise).
For each terminal $(x,0)$ with parity $1\oplus s$,
set $\sigma(x,0)=1$ with probability~$\pneq$ (and  set $\sigma(x,0)=0$ otherwise).

Informally, the distribution~$\muall^s$ will 
be relevant when an idealised gadget is in a phase
which prefers $1$-spins at  parity-$s$ terminals. 
In this distribution,
the probability that a terminal is given spin~$1$
is higher if the terminal  has parity~$s$ than if it has parity~$1\oplus s$.
  
Let $\muall$ be the distribution on configurations $\sigma:\configs{\terminals}$
given by
$\muall(\sigma) = (\mueven(\sigma)+\muodd(\sigma))/2$.
We will show that, provided that   $d$  is sufficiently large,
the distribution of $\sigmarv(\terminals)$ is close to $\muall$.

Thus, the gadget  
can be thought of informally as having two  phases, 
phases~$0$ and~$1$. We will show that the gadget 
almost always occupies one of these two phases, and they 
occur with equal probability.
In   phase~$0$, the distribution of $\sigmarv(\terminals)$ is close to $\mueven$.
In  phase~$1$, the distribution of $\sigmarv(\terminals)$ is close to $\muodd$.

 \begin{proposition} 
 \label{prop:gadget}  
There is a $c>1$ such that,
if $d$ is a sufficiently large multiple of~$16$,
$k$ is an integer greater than or equal to~$1$ and
$\tau$ is any configuration $\tau:\configs{\terminals}$, then
$$|\Pr(\sigmarv(\terminals)=\tau)-\muall(\tau)| \leq c^{-d} k^2.$$ 
  \end{proposition}
 
Proposition~\ref{prop:gadget} is established at the end of this
section. We will use contour arguments adapted from Dobrushin~\cite{Dob68} and Borgs et al.~\cite{Borgs}.
The outline of the argument is as follows.  We first define
``contours''  in Section~\ref{sec:contours}. We  show, in
Section~\ref{sec:notlong}, that
long contours are unlikely.  
In Section~\ref{sec:nearind}, we show that, in the absence of long contours,
the spins of terminals are nearly independent.
With high probability,   the gadget has a phase~$s$ and there is a boundary around
each terminal, whose spins are consistent with~$s$. Conditioned
on~$s$, the distribution of the spins of the terminals is close to $\muall^s$.
 
\subsection{The Dual Gadget, trails, and contours}
\label{sec:contours}

The dual gadget~$C^*_\size$  
has vertex set
$V(C^*_\size) = \DualColumns \times \DualRows$. 
Vertices $(x,y)$ and $(x',y')$ are adjacent in~$C^*_\size$ if
\begin{itemize}
\item $y=y'$ and 
$y\notin\{-\tfrac{1}{2},\size+\tfrac{1}{2}\}$, and
$x=x'\pm 1$ (where of course, the arithmetic is modulo~$2\size$), 
or
\item $x=x'$ and $y=y' \pm 1$.
\end{itemize}
$E(C^*_\size)$ denotes the edge set of $C^*_\size$.
This is illustrated in Figure \ref{fig:dualgraph}.

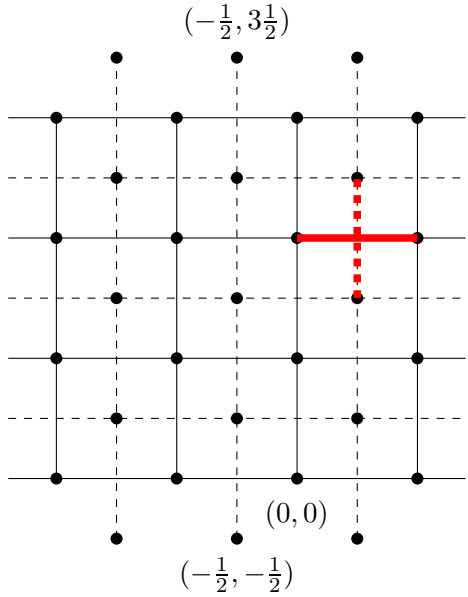
\begin{figure}[!h] 
\begin{center}
\begin{tikzpicture}[scale=1.6]
% Primal
\draw (-2.4,0) -- (1.4,0) (-2.4,1) -- (1.4,1) (-2.4,2) -- (1.4,2) (-2.4,3) -- (1.4,3);
\draw (-2,0) -- (-2,3) (-1,0) -- (-1,3) (0,0) -- (0,3) (1,0) -- (1,3);
\foreach \x in {-2,-1,0,1}
 \foreach \y in {0,1,2,3}
   \fill (\x,\y) circle (0.05);
% Dual
\draw[dashed] (-2.4,0.5) -- (1.4,0.5) (-2.4,1.5) -- (1.4,1.5) (-2.4,2.5) -- (1.4,2.5);
\draw[dashed] (-1.5,-0.5) -- (-1.5,3.5) (-0.5,-0.5) -- (-0.5,3.5) (0.5,-0.5) -- (0.5,3.5);
\foreach \x in {-1.5,-0.5,0.5}
 \foreach \y in {-0.5,0.5,1.5,2.5,3.5}
   \fill (\x,\y) circle (0.05);
\draw (-0.5,3.8) node {$(-\tfrac 1 2,3\tfrac 1 2)$};
\draw (-0.5,-0.8) node {$(-\tfrac 1 2,-\tfrac 1 2)$};
\draw (0,-0.3) node {$(0,0)$};

\draw[red, line width=0.1cm, dashed] (0.5,1.5) -- (0.5,2.5);
\draw[red, line width=0.1cm] (0,2) -- (1,2);
\end{tikzpicture}
\end{center}
\caption{Part of $C_3$ and $C_3^*$; solid lines are edges of $C_3$,
 dashed lines are edges of $C_3^*$. The red thickened lines are a
 dual pair of edges.}
 \label{fig:dualgraph}
\end{figure}

There is a bijection called ``duality'' 
between edges of $C_\size$ and
edges of $C^*_\size$.
In particular, the dual of edge  $e=((x,y),(x+1,y))$ of $C_\size$ is  
$e^* = ((x+\tfrac 1 2,y-\tfrac 1 2),(x+\tfrac 1 2,y+\tfrac 1 2))$
and the dual of edge~$e^*$ is~$e$.
Similarly, the dual of edge $f=((x,y),(x,y+1))$ of $C_\size$ is
$f^* = ((x-\tfrac12,y+\tfrac12),(x+\tfrac12,y+\tfrac12))$ and the dual of~$f^*$ is~$f$.
We use the $^*$-operation to move between an edge and its dual, 
so every edge~$e$ satisfies ${(e^*)}^* = e$.

\begin{figure}[!h]
\def\hook#1#2#3#4{(canvas polar cs:angle=10*#1,radius=2cm+0.2*#2cm) arc(10*#1:10*#3:2cm+0.2*#2cm) --
(canvas polar cs:angle=10*#3,radius=2cm+0.2*#4cm)}
\begin{center}
\begin{tikzpicture}[scale=0.9]
\foreach \x in {0,...,35} \draw
  (canvas polar cs:angle=-90+10*\x,radius=1.8cm) --
  (canvas polar cs:angle=-90+10*\x,radius=4.2cm);
\foreach \y in {0,...,10} \draw (0,0) circle (4cm-0.2*\y cm);
\foreach \x in {0,...,35} \foreach \y in {0,...,10}
    \fill[gray] (canvas polar cs:angle=-90+10*\x,radius=4cm-0.2*\y cm) circle (0.04cm);
\draw  [violet, ultra thick]  \hook{17}{6}{20}{9};
\draw  [violet, ultra thick]  \hook{20}{9}{19}{8};
\draw  [violet, ultra thick]  \hook{19}{8}{17}{6};

\draw  [violet, ultra thick]  \hook{5}{8}{8}{11};
\draw  [violet, ultra thick]  \hook{5}{6}{9}{7};
\draw  [violet, ultra thick]  \hook{5}{6}{5}{8};
\draw  [violet, ultra thick]  \hook{9}{7}{10}{11};

\draw  [violet, ultra thick]  \hook{27}{-1}{27}{3};
\draw  [violet, ultra thick]  \hook{27}{3}{25}{5};
\draw  [violet, ultra thick]  \hook{25}{5}{30}{2};
\draw  [violet, ultra thick]  \hook{30}{2}{29}{-1};

\draw  [violet, ultra thick]  \hook{2}{3}{8}{2};
\draw  [violet, ultra thick]  \hook{8}{2}{12}{5};
\draw  [violet, ultra thick]  \hook{12}{5}{16}{3};
\draw  [violet, ultra thick]  \hook{16}{3}{23}{8};
\draw  [violet, ultra thick]  \hook{23}{8}{32}{6};
\draw  [violet, ultra thick]  \hook{32}{6}{38}{3};

\draw  [blue, dashed, ultra thick]  \hook{33}{-1}{33}{4};
\draw  [blue, dashed, ultra thick]  \hook{33}{4}{35}{9};
\draw  [blue, dashed, ultra thick]  \hook{35}{9}{34}{11};

 \draw (canvas polar cs:angle=329,radius=1.6cm) node {$C$};
 \draw (canvas polar cs:angle=341,radius=4.5cm) node {$C'$};
\end{tikzpicture}
\end{center}
\caption{Some representative simple contours (violet, solid) and a single cross contour (blue, dashed, from $C$ to $C'$)
in the dual gadget.}
\label{fig:contours}
\end{figure}
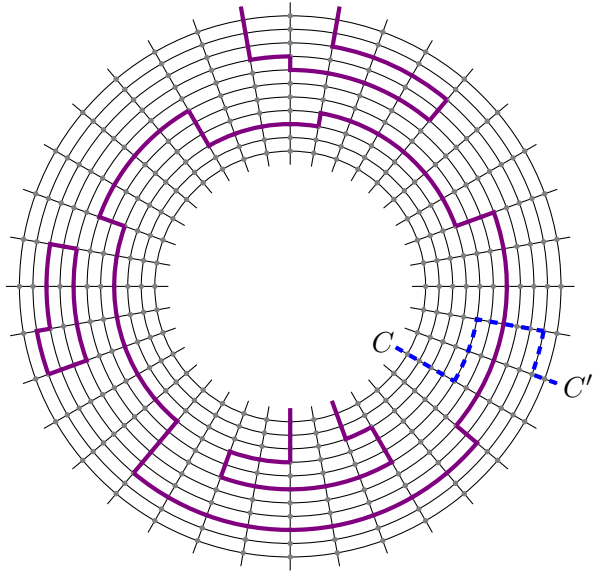

A \emph{trail} in $C^*_\size$ is a sequence $g=v_1,\dots,v_j$
of vertices in~$V(C^*_\size)$ such that 
each pair~$(v_i,v_{i+1})$ is an edge of $C^*_\size$, and no edge
is used twice.   
A \emph{contour} is a trail $g=v_1,\dots,v_j$ in $C^*_\size$ satisfying one of the following:
\begin{itemize}
\item $v_1=v_j$, or
\item The $y$-coordinate of~$v_1$ and the $y$-coordinate of~$v_j$
are both in $\{-\tfrac{1}{2},\size+\tfrac{1}{2}\}$.
\end{itemize}
The \emph{length} of $g$ is $j-1$.
We say that~$g$ is a \emph{cross contour} if the $y$-coordinate of~$v_1$
is~$-\frac12$ and the $y$-coordinate of~$v_j$ is~$\size+\frac12$ (or vice-versa). 
A cross contour goes from one boundary of the gadget to the other.
We say that every other contour  is a \emph{simple} contour. Refer to Figure~\ref{fig:contours},
where some contours are visualised.
 
Given $\sigma:\configs{V(C_\size)}$,
let $\sigma^*$ be the set of edges of $C^*_\size$
which are dual to monochromatic edges.
In particular,
$\sigma^* = \{  (u,v)^* \in E(C^*_\size) \mid 
 \sigma(u)=\sigma(v) \}$. 
A contour of~$\sigma$ (to be defined presently) is a trail in the dual gadget that 
separates vertices 
with parity-$0$ ones from vertices with parity-$1$ ones in the (primal) gadget
(the first bullet point in Definition~\ref{def:contours}).  In addition, the    
parities of ones immediately to the ``left'' and ``right'' of the contour
are consistent along its entire length (the second bullet point).
 \begin{definition}\label{def:contours}
 A \emph{contour} of~$\sigma$ is a contour~$g=v_1,\ldots,v_j$ satisfying the following two properties.
 \begin{itemize}
 \item The edges of~$g$ are monochromatic: That is, 
 for all $1\leq i < j$, $(v_i,v_{i+1})\in \sigma^*$.
 \item The contour~$g$ always turns at degree-$4$ vertices: That is,
 for all $1< i < j$,
 if four edges of $\sigma^*$ meet at vertex~$v_i$, then
 $v_{i-1}$ and $v_{i+1}$ differ in both the $x$ component and the $y$ component.
 Similarly, if four edges of $\sigma^*$ meet at $v_1=v_j$ then $v_2$ and $v_{j-1}$ differ
 in both the $x$ component and the $y$ component.
 \end{itemize}
 \end{definition}
  
 Note that contours of~$\sigma$ cannot cross each other, though two contours
 can  share a vertex without crossing. Also, two contours can have a common
 portion (before turning off in two different directions).
 Finally, every edge of $\sigma^*$ is contained in at least one contour of $\sigma$.

 Let $\sigma:\configs{V(C_\size)}$ be a configuration, and let 
 $g$ be a contour of~$\sigma$.
 We say that a vertex $u\in V(C_\size)$ is adjacent to~$g$
 if there is an edge $(u,v)\in E(C_\size)$ such that $e^*\in g$.
 The set of vertices adjacent to~$g$ can
 be written as the union of two sets, $L(g)$ and $R(g)$, where
 $L(g)$ is the set of vertices of $V(C_\size)$ that are on the left (relative to the direction of travel)
 when we follow the trail~$g$ from~$v_1$ to~$v_j$, and
 $R(g)$ is the set of vertices of $V(C_\size)$ that are on the right (relative to the direction of travel).
 See Figure~\ref{fig:two}.
 \begin{figure}[!h]
\begin{center}
\begin{tikzpicture}
\draw (-0.3,-0.3) grid (5.3,5.3);
\draw[->-] (1,2) -- (1,3);
\draw[->-] (1,3) -- (2,3);
\draw[->-] (2,3) -- (2,4);
\draw[->-] (2,4) -- (3,4);
\draw[->-] (3,4) -- (3,3);
\draw[->-] (3,3) -- (4,3);
\draw[->-] (4,3) -- (4,2);
\draw[->-] (4,2) -- (4,1);
\draw[->-] (4,1) -- (3,1);
\draw[->-] (3,1) -- (2,1);
\draw[->-] (2,1) -- (2,2);
\draw[->-] (2,2) -- (1,2);
\foreach \x/\y in {0/2,1/1,1/3,2/0,2/4,3/0,3/3,4/1,4/2} \draw (\x+0.5,\y+0.5) node {L};
\foreach \x/\y in {1/2,2/1,2/3,3/1,3/2} {
  \fill[shade] (\x,\y) rectangle +(1,1);
  \draw (\x+0.5,\y+0.5) node {R};
}
\end{tikzpicture}
\end{center}
\caption{Left and right vertices of a contour of $\sigma$.  The shaded
  squares represent vertices of $C_\size$ with parity-$1$ ones and the
  unshaded squares represent vertices of $C_\size$ with parity-$0$ ones.}
\label{fig:two}
\end{figure}
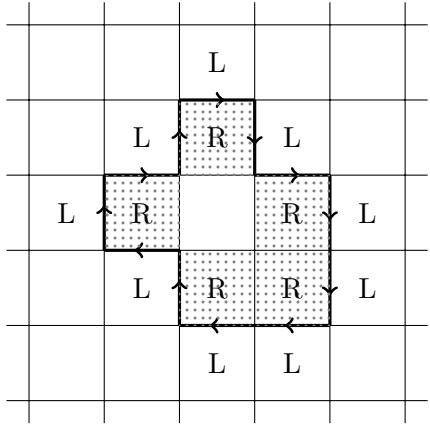
  
A key property of contours is the following.
 
 \begin{lemma}
 \label{lem:contourboundary}
 Let $\sigma:\configs{V(C_\size)}$ be a configuration, and let~$g$
 be a contour of~$\sigma$.
 Then for some $s\in\{0,1\}$,
  $\sigma(L(g))$ has  parity-$s$ ones and $\sigma(R(g))$ has  parity-$(1{\oplus}s)$ ones.
   \end{lemma}

   \begin{proof}
Pick $s\in\{0,1\}$ such that the vertex on the left as we go from
$v_1$ to $v_2$ has parity-$s$ ones.
By induction on $i$, we will show that for each $i$ the vertex on
the left as we go from $v_i$ to $v_{i+1}$ has parity-$s$ ones.
Suppose without loss of generality that the edge $(v_{i-1},v_i)$ increases the $x$-component
(the other three cases are similar). So
$v_{i-1} =  (x-\frac12,y+\frac12)$
and $v_i =  (x+\frac12,y+\frac12)$.
Since $g$ is a contour of~$\sigma$,  $\sigma(x,y)=\sigma(x,y+1)=s\oplus x\oplus y$. There are three cases.

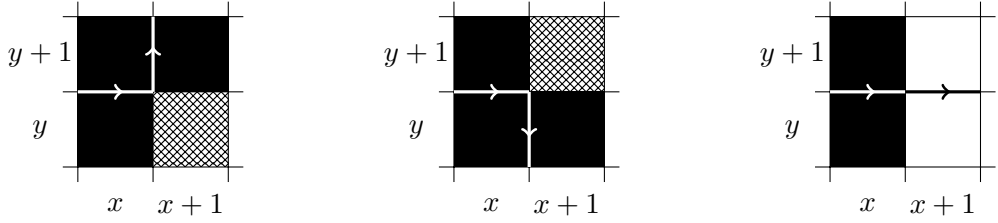
\begin{figure}[!h]
\begin{center}
\begin{tikzpicture}
\newcommand{\drawstep}[2]{
 \begin{scope}[xshift=#1 cm]
   \draw (-0.2,-0.2) grid (2.2,2.2);
 \draw (0.5,-0.5) node {$x$}  (1.5,-0.5) node {$x+1$};
 \draw (-0.5,0.5) node {$y$}  (-0.5,1.5) node {$y+1$};
 \fill[black] (0,0) rectangle +(1,1);
 \fill[black] (0,1) rectangle +(1,1);
 \draw[white, ->-] (0,1) -- (1,1);
   #2
 \end{scope}
}
\drawstep{0}{
 \fill[pattern=crosshatch] (1,0) rectangle +(1,1);
 \fill[black] (1,1) rectangle +(1,1);
 \draw[white, ->-] (1,1) -- (1,2);
}
\drawstep{5}{
 \fill[black] (1,0) rectangle +(1,1);
 \fill[pattern=crosshatch] (1,1) rectangle +(1,1);
 \draw[white, ->-] (1,1) -- (1,0);
}
\drawstep{10}{
 \draw[black, ->-] (1,1) -- (2,1);
}
\end{tikzpicture}
\end{center}
\caption{Cases 1, 2, and 3. Black squares have the same spin as
 $\sigma(x,y)$; white squares have the opposite spin, and hatched
 squares can be either.}
\end{figure}

\begin{enumerate}
\item $v_{i+1}=(x+\frac12,y+\frac32)$.  In this case the vertex $(x,y+1)$ is
 still on the left as we go from $v_i$ to $v_{i+1}$.
\item $v_{i+1}=(x+\frac12,y-\frac12)$.  In this case the vertex $(x+1,y)$ is
 on the left as we go from $v_i$ to $v_{i+1}$, but since $g$ is a
 contour of $\sigma$ we have
 $\sigma(x+1,y)=\sigma(x,y)=\sigma(x,y+1)$. So $(x+1,y)$ has
 parity-$s$ ones.
\item $v_{i+1}=(x+\frac32,y+\frac12)$.  In this case $(x+1,y+1)$ is on the
 left, and $(x+1,y)$ is on the right.  Since $g$ is a contour of
 $\sigma$, $\sigma(x+1,y)=\sigma(x+1,y+1)$, and we know
 $\sigma(x,y)=\sigma(x,y+1)$.  Since the contour did not turn, the
 vertex $(x+\frac12,y+\frac12)$ cannot have degree $4$ in $\sigma^*$, so
 $\sigma(x+1,y+1)=s\oplus x\oplus y\oplus 1$, so $(x+1,y+1)$ has parity-$s$
 ones. \qedhere
\end{enumerate}
\end{proof}

The following lemma allows $w_{G}(\sigma)$ to be expressed more
easily in terms of the contours of $\sigma$.  
Suppose $\size > 2$. A \emph{side vertex} of~$C_\size$ is a vertex $(x,y)\in V(C_\size)$ with 
$y=0$ or $y=\size$. A~\emph{side edge} is an edge 
in $E(C_\size)$ between two side vertices.
Recall that $b(\sigma)$ (respectively, $c(\sigma)$)
denotes the number of edges~$(u,v)$ of~$G$
with $\sigma(u)=\sigma(v)=0$
(respectively, $\sigma(u)=\sigma(v)=1$),
and $\ell(\sigma)$ denotes the number of vertices~$u$
with $\sigma(u)=1$.
 
\begin{lemma}\label{lem:ell}
Fix $\size>2$ and a configuration $\sigma:\configs{V(C_\size)}$.  Let
$b'(\sigma)$ be the number of side  edges~$(u,v)$ of~$C_\size$
with $\sigma(u)=\sigma(v)=0$
and let $c'(\sigma)$ be the number of side edges~$(u,v)$ of~$C_\size$
with $\sigma(u)=\sigma(v)=1$.  Then
\begin{align*}
\ell(\sigma)=\tfrac{1}{4}(c(\sigma)-b(\sigma))+\tfrac{1}{8}(c'(\sigma)-b'(\sigma))+\size(\size+1).
\end{align*}
\end{lemma}
\begin{proof}
Let $\ell'(\sigma)$ be the number of side vertices~$u$ with $\sigma(u)=1$. Let  $E'$ be the set of all side edges in $C_\size$.  By
double-counting pairs $(u,(u,v))$ with $\sigma(u)=1$ and $(u,v)\in E(C_\size)$,
\begin{align*}
(|E(C_\size)|-b(\sigma)-c(\sigma))+2c(\sigma)&=4\ell(\sigma)-\ell'(\sigma).
\end{align*}
By double-counting pairs $(u,(u,v))$ with $\sigma(u)=1$ and $(u,v)\in E'$,  we have
\begin{align*}
(|E'|-b'(\sigma)-c'(\sigma))+2c'(\sigma)&=2\ell'(\sigma).
\end{align*}
Rearranging gives
$\ell(\sigma)=\tfrac{1}{4}(c(\sigma)-b(\sigma))+\tfrac{1}{8}(c'(\sigma)-b'(\sigma))+\tfrac
1 4 |E(C_\size)| + \tfrac 1 8 |E'|$.  
Consider the
configuration 
with alternating~$0$s and~$1$s given by
$\sigma(x,y)=x\oplus y$. For this configuration~$\sigma$, we have $b(\sigma)=c(\sigma)=b'(\sigma)=c'(\sigma)=0$
and $\ell(\sigma)=\size(\size+1)$, so the constant term $\tfrac 1 4 |E(C_\size)| +
\tfrac 1 8 |E'|$ is $\size(\size+1)$.
\end{proof}

  \subsection{Long contours are unlikely}
  \label{sec:notlong}

 \begin{lemma}\label{lem:longsimple_unlikely} 
There is a $c>1$ such that, 
for all sufficiently large~$h$, 
all $\size>2$, and all $U\subseteq V(C^*_\size)$,
$$\Pr(\mbox{$\bsigmasize$ has a simple contour of length at least~$h$
starting in~$U$} )\leq |U| \, c^{-h}.$$
 \end{lemma}
\begin{proof} 

Suppose that $g$ is a simple length-$r$ contour of a configuration
$\sigma\colon\configs{V(C_\size)}$. 
Consider the connected components of the graph
$\big(V(C_\size),E(C_\size)\setminus\{e^* \mid e \in g\}\big)$.
We say that a component is ``left''
if it contains at least one vertex in $L(g)$ (but no vertices in $R(g)$).
We say that it is ``right'' if it
contains at least one vertex in $R(g)$ (but no vertices in $L(g)$).
Every component is either left or right.
Let $S$ be the set of vertices in left components. 
Let $\overline S = V(C_\size) \setminus S$.
Let $S' = \{ (x,y) \in   S \mid (x-1,y)\in \overline S\}$,
where,  as usual, the arithmetic on~$x$ is done modulo~$2 \size$.

Suppose that $\sigma(R(g))$ has parity-$s$ ones.
By Lemma~\ref{lem:contourboundary}, this is true for some $s\in\{0,1\}$.
Define a configuration $\sigma^g:\configs{V(C_\size)}$ as follows:
$\sigma^g(\overline S) = \sigma(\overline S)$,
$\sigma^g(S')$ has parity-$s$ ones,
and, for every $(x,y) \in S\setminus S'$, 
$\sigma^g( x,y ) = \sigma(x-1,y)$.  (Informally, $\sigma^g$ is obtained
by shifting the configuration on~$S$ one place to the right, i.e., in 
the direction of increasing~$x$, and assigning consistent parity to
the vacated vertices.)

Note that $\sigma\mapsto \sigma^g$ is a map from the set of configurations~$\sigma$
with $g$ as a contour to the set of all configurations;  further, it does not lose 
information, and hence is injective.  (To see that the map is injective, note that 
on most of~$S$ the configuration~$\sigma$ may be recovered 
using the identity $\sigma(x,y)=\sigma^g(x+1,y)$.
The vertices with $(x,y)\in S$ and $(x+1,y)\in \overline S$ are exceptional,
but they are adjacent to the contour $g$ and hence $\sigma$ is determined on these.
Of course, $\sigma=\sigma^g$ on $\overline S$.) 

Note also that $(\sigma^g)^*$ is the same as $\sigma^*$, but with $g$ removed and
with the contours in $S$ shifted by one.
By Lemma \ref{lem:ell},
\begin{align*}
w_{C_{\size}}(\sigma)&=(\beta\lambda^{-1/4})^{b(g)}(\gamma\lambda^{1/4})^{c(g)}
   \lambda^{(c'(g)-b'(g))/8}w_{C_{\size}}(\sigma^g),
\end{align*}
where $b(g),c(g),b'(g),c'(g)$ are the contributions to
$b(\sigma),c(\sigma),b'(\sigma),c'(\sigma)$ coming from edges 
whose duals are
in~$g$.
As the map $\sigma\mapsto \sigma^g$ is injective,
\begin{align*}
\Pr(g\subseteq\bsigmasize^*) &\leq (\beta\lambda^{-1/4})^{b(g)}(\gamma\lambda^{1/4})^{c(g)}\lambda^{(c'(g)-b'(g))/8}\\
&\leq (\beta\lambda^{-1/4})^{b(g)}(\gamma\lambda^{3/8})^{c(g)},
\end{align*}
where we have used the facts that $\lambda\geq1$ and $c'(g)\leq c(g)$.
There are at most $|U|\,3^r$ relevant contours of length~$r$ in total
($|U|$ choices of starting point, 
and at most three different directions at each step), so 
\begin{align*}
&\Pr(\text{$\bsigmasize$  has a simple contour of length at least $h$ starting in $U$})\\
&\qquad\null\leq | U|\sum_{r\geq h}3^r\max(\beta\lambda^{-1/4},\gamma\lambda^{3/8})^r\\
&\qquad\null\leq | U |\,\frac{(3\times0.238)^{h}}{1-3\times0.238} .\\
\end{align*} 
There is a $c>1$ such that $\frac{(3\times0.238)^{h}}{1-3\times0.238}< c^{-h}$ for
all sufficiently large~$h$.
\end{proof}

 \begin{lemma}\label{lem:longsimple_unlikelycond}  
Let $i\in\{0,1\}$.
 For every  $\size>2$, and every simple contour~$g$ of length~$r$,
\begin{align}\label{eq:longsimple_unlikelycond_single}
\Pr(\mbox{$g$ is a contour of $\bsigmasize$ }
\mid   \mbox{$\bsigmasize(B_{i, \size})$ has parity-$0$ ones}
)\leq  \max(\beta \lambda^{-1/4}, \gamma \lambda^{3/8} )^r.
\end{align}
Furthermore, there is a $c>1$ such that, for all   sufficiently large $h$, all $\size>2$,
and all $U\subseteq V(C_\size)$,
the conditional probability that $\bsigmasize$ has a simple contour of length at least~$h$ 
which contains an edge whose dual connects two vertices in $U$,
conditioned on the fact that $\bsigmasize(B_{i, \size})$ has parity-$0$ ones,
is at most $ | U|\,c^{-h}$.
\end{lemma} 
\begin{proof} 
The proof of \eqref{eq:longsimple_unlikelycond_single} is similar to the first half of the proof of Lemma~\ref{lem:longsimple_unlikely}, except that 
we have to take care to choose~$S$ to be on the correct side of the contour~$g$.
Previously, it did not matter whether we formed $S$ from the left or right components,
and we arbitrarily chose the former.  Now we choose $S$ (either taking
all the left or all the right components) in such a way that $S\cap B_{i,\size}=\emptyset$.
This is possible because 
all the vertices in $B_{i,\size}$ are in a single connected 
component (the contour $g$ does not cross any   edges whose endpoints lie in $B_{i,\size}$).
Now define $\sigma^g$ as in the proof of Lemma~\ref{lem:longsimple_unlikely} and continue as before. This establishes \eqref{eq:longsimple_unlikelycond_single}.

For all $1\leq s\leq r$, and all $u\in U$, there are at most $3^r \times 4$
contours $v_1\dots v_r$ for which $u$ is on the left as we go from
$v_{s-1}$ to $v_s$: a choice of initial direction and direction at each step determines the
contour.  Summing over $s$ and $u$, this implies that there are at
most $4|U|r3^r$ length-$r$ contours with an edge whose dual connects vertices of
$U$. By \eqref{eq:longsimple_unlikelycond_single},
\begin{align*}
&\Pr\left(\text{\parbox{3.7in}{\centering{$\bsigmasize$ has a simple contour of length at least~$h$
which contains an edge whose dual connects two vertices in $U$}}} \;\middle\vert\;
\text{$\bsigmasize(B_{i, \size})$ has parity-$0$ ones}\right)\\
&\qquad\null\leq 4|U|\sum_{r\geq h}r 3^r\max(\beta\lambda^{-1/4},\gamma\lambda^{3/8})^r\\
&\qquad\null= 4|U|(3\times 0.238)^h\sum_{t\geq 0}(t+h)(3\times 0.238)^t\\
&\qquad\null= 4|U|(3\times 0.238)^h\,\frac{3\times0.238+h(1-3\times 0.238)}{(1-3\times 0.238)^2}.
\end{align*}
There is a $c>1$ such that $4(3\times 0.238)^h\,\frac{3\times0.238+h(1-3\times 0.238)}{(1-3\times 0.238)^2}< c^{-h}$ for
all sufficiently large~$h$.
\end{proof}

By the {\it upper boundary of $C_\size$} we mean the set of all 
vertices of the form $(x,\size)$ for some~$x$.     
\begin{lemma}\label{lem:newcontour}  
Let $i\in \{0,1\}$.
   There is a $c>1$ such that, for all  sufficiently large $h$ and all 
$\size>h$,
   the probability that
     $\bsigmasize$ has a simple contour 
     that separates  
     the set 
$\{-h+i,\ldots,h+i\} \times \{0,\ldots,h\}$ 
from the upper boundary of $C_\size$, conditioned on the fact that 
$\bsigmasize(B_{i, \size})$ has parity-$0$ ones,
 is at most $c^{-h}$.
 \end{lemma}  
 
\begin{proof}
Note that the separating contour cannot wrap around, owing to the 
boundary condition that $\bsigmasize(B_{i, \size})$ has parity-$0$ ones.
If the separating contour has length~$r+2$ then its right-endpoint is
in the  set
$\{(h+i+x-\frac12,-\frac12) \mid 1\leq x \leq r\}$.
There is a unique choice for the edge incident to each endpoint. Thus, there
are at most $r3^r$ possible contours. By Lemma~\ref{lem:longsimple_unlikelycond},
the probability that $\bsigmasize$ has such a simple contour,
conditioned on the fact that $\bsigmasize(B_{i,\size})$ has parity-0 ones,
is at most 
$$\sum_{r=h-2}^\infty r3^r \max(\beta \lambda^{-1/4}, \gamma \lambda^{3/8})^{r+2}.$$
Thus, the probability is at most
$$\max(\beta\lambda^{-1/4},\gamma\lambda^{3/8})^2\sum_{r=h-2}^\infty
r 3^r\max(\beta\lambda^{-1/4},\gamma\lambda^{3/8})^r,$$
which, as in the proof of Lemma~\ref{lem:longsimple_unlikelycond}, is exponentially small in $h$.
\end{proof}
  
  \begin{lemma}\label{lem:longcross_unlikely} 
There is a $c>1$ such that,  
for all sufficiently large~$\size$,
$$\Pr(\mbox{$\bsigmasize$ has a  cross contour  })\leq c^{-\size}.$$
\end{lemma} 
\begin{proof}
Let $g$ be a cross contour.  There must be at least one other 
cross contour~$g'$.  For otherwise there would be a path $p$ in $C_{\size}$ from $L(g)$ to $R(g)$ 
such that $\bsigmasize(V(p))$ has parity-$0$ ones or parity-$1$ ones, which would violate parity.
Orient $g$ and~$g'$ in opposite senses (one away from $y=-\frac12$ and one towards).
Consider the connected components of the graph
$\big(V(C_\size),E(C_\size)\setminus\{e^* \mid e \in g\cup g'\}\big)$,
and let $S$ be the union of all connected components that are left of either $g$ or~$g'$.
Now proceed as in the proof of Lemma~\ref{lem:longsimple_unlikely},
using the fact that a cross contour has length at least $\size$ and the
set of possible starting points has size~$2\size$.
\end{proof}

\begin{lemma}\label{lem:longcross_unlikelycond}
Let $i\in \{0,1\}$.
Fix $\size\geq 1$.
Conditioned on $\bsigmasize(B_{i,\size})$ having parity-$s$ ones
(for any $s\in\{0,1\}$),
$\bsigmasize$ has no cross contour.
\end{lemma}

 \begin{proof}
A cross contour would have to cross a side edge in $B_{i,\size}$, which is impossible.
\end{proof}

\begin{lemma}\label{lem:peqgreater}
$\peq > \pneq$.
\end{lemma}
\begin{proof}  
Fix $\size>2$.
Suppose that $\bsigmasize(B_{0,\size})$ has parity-$0$ ones.
If $\bsigmasize(0,0)=0$ then there is a simple contour of~$\bsigmasize$ that separates 
$(0,\size)$ from $(0,0)$.  (Note that, by Lemma~\ref{lem:longcross_unlikelycond},
cross contours cannot separate these two vertices.)
If the separating contour has length~$r+2$ then its right-endpoint is
in the range $( \frac12,-\frac12),\ldots,(r-\frac12,-\frac12)$.
There is a unique choice for the edge incident to each endpoint. Thus, there
are at most $r3^r$ possible contours. By Lemma~\ref{lem:longsimple_unlikelycond},
the probability that $\bsigmasize$ has such a simple contour,
conditioned on the fact that $\bsigmasize(B_{0,\size})$ has parity-0 ones,
is at most 
$$\sum_{r=1}^\infty r3^r \max(\beta \lambda^{-1/4}, \gamma \lambda^{3/8})^{r+2}.$$
Thus, 
\begin{align*}
 \Pr(\bsigmasize(0,0)=0 \mid 
\mbox{$\bsigmasize(B_{0,\size})$ has parity-0 ones}) 
&\leq\max(\beta\lambda^{-1/4},\gamma\lambda^{3/8})^2\sum_{r=1}^\infty
r 3^r\max(\beta\lambda^{-1/4},\gamma\lambda^{3/8})^r\\
&\leq(0.238)^2\frac{3\times 0.238}{(1-3\times0.238)^2}\\
&< 1/2. 
\end{align*}
Thus, $\peq>\frac12$.

Similarly, suppose that $\bsigmasize(B_{1,\size})$ has parity-0 ones.
If there is no simple contour of~$\bsigmasize$ that separates 
$(1,\size)$ from $(1,0)$, then $\bsigmasize(1,0)=0$.
We already saw that the probability that no such contour exists is
greater than $\frac12$.
Thus, 
$$ 
\Pr(\bsigmasize(1,0)=0 \mid 
\mbox{$\bsigmasize(B_{1,\size})$ has parity-0 ones}) >1/2.$$ 
So $\pneq<\frac12$.
Putting the two inequalities together, we have $\pneq<\frac12<\peq$.
\end{proof}

\subsection{In the absence of long contours, the spins of the terminals are nearly independent}
 \label{sec:nearind}
The main goal of this section is to establish Proposition~\ref{prop:gadget}.
Although the details are technical, the underlying idea is
straightforward.  For the range of parameters $(\beta,\gamma,\lambda)$ we are considering,
we expect a typical configuration~$\sigma$ to have a well-defined phase, i.e.,
either a substantial majority of vertices will 
have parity~$0$ ones, or a substantial majority will have parity~$1$ ones. 
This is the ``phase'' of the gadget.  Around each terminal we locate a connected
region of $V(C_\size)$, a ``$d$-boundary'', that is not too close to the terminal, 
that has ones of constant parity, and that separates the terminal from the rest of the gadget.  
Since the $d$-boundaries are far from
each terminal, they are large and  the parities of their ones will, with overwhelming probability,
be consistent among themselves.
Further, by re-randomising $\sigma$ inside the $d$-boundaries, we can appreciate that 
the spins at the terminals are conditionally nearly independent given the phase; 
moreover, since the $d$-boundaries are far from their respective terminals,
the distributions of spins at the terminals are nearly identical.
We now make these ideas precise.

The \emph{$*$-distance} between two points $(x,y)$ and $(x',y')$ in~$V(C_\nu)$
is $\max(|x-x'|,|y-y'|)$ where $|x-x'|$ is the minimum non-negative
integer such that $x=x'+|x-x'|$ modulo $2\nu$ or $x'=x+|x-x'|$ modulo
$2\nu$.   
The \emph{$*$-distance} between two points $(x,y)$ and $(x',y')$ in~$V(C^*_\nu)$
is defined similarly.
Let $U_{x,h}$ be the set of vertices of~$C_\size$ whose $*$-distance from $(x,0)$ 
is at most~$h$.
Vertices in $V(C_\size)$ (or $V(C_\size^*)$)
 are \emph{$*$-adjacent}
if the $*$-distance between them is~$1$.
A~$*$-path on $V(C_\size)$ 
is a sequence $v_1,\ldots,v_h$ 
of vertices in $V(C_\nu)$ such that, for each $j \in \{1,\ldots,h-1\}$,
the vertices $v_j$ and $v_{j+1}$ are $*$-adjacent. 
A~$*$-path on $V(C_\size^*)$ is defined similarly.

\begin{definition}\label{def:boundaries}
Suppose  $\size$ and $h$ are positive integers. 
An \emph{$h$-boundary} of a vertex $(x,0)$
is a set of vertices $B\subseteq V(C_\size)$ such that 
the following are true.
 \begin{enumerate}[label=(\roman{*})]
\item $(x,0)$ is not connected to $(x,\size)$ in 
the graph $C_\size \setminus B$, and
\item $B$ does not intersect $U_{x,h/4}$, and
\item $B$ is a subset of $U_{x,h/2}$.
\item The subgraph of $C_\size$ induced by $B$ is connected.
\end{enumerate}
See Figure~\ref{fig:three}.
\end{definition}

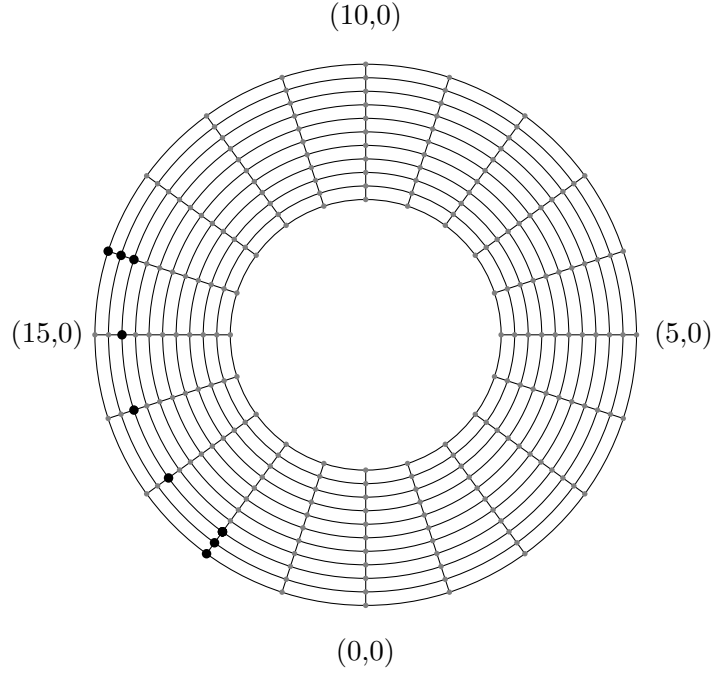
\begin{figure}[!h]
\begin{center}
\begin{tikzpicture}[scale = 0.9]
\foreach \x in {0,...,19} \draw
  (canvas polar cs:angle=-90+18*\x,radius=2cm) --
  (canvas polar cs:angle=-90+18*\x,radius=4cm);
\foreach \y in {0,...,10} \draw (0,0) circle (4cm-0.2*\y cm);
\foreach \x in {0,...,19} \foreach \y in {0,...,10}
    \fill[gray] (canvas polar cs:angle=-90+18*\x,radius=4cm-0.2*\y cm) circle (0.04cm);
\foreach \x in {16} {
    \fill
(canvas polar cs:angle=-90+18*\x-2*18,radius=4cm-0*0.2cm) circle (0.07cm)
(canvas polar cs:angle=-90+18*\x-2*18,radius=4cm-1*0.2cm) circle (0.07cm)
(canvas polar cs:angle=-90+18*\x-2*18,radius=4cm-2*0.2cm) circle (0.07cm)
(canvas polar cs:angle=-90+18*\x-1*18,radius=4cm-2*0.2cm) circle (0.07cm)
(canvas polar cs:angle=-90+18*\x-0*18,radius=4cm-2*0.2cm) circle (0.07cm)
(canvas polar cs:angle=-90+18*\x+1*18,radius=4cm-2*0.2cm) circle (0.07cm)
(canvas polar cs:angle=-90+18*\x+2*18,radius=4cm-2*0.2cm) circle (0.07cm)
(canvas polar cs:angle=-90+18*\x+2*18,radius=4cm-2*0.2cm) circle (0.07cm)
(canvas polar cs:angle=-90+18*\x+2*18,radius=4cm-1*0.2cm) circle (0.07cm)
(canvas polar cs:angle=-90+18*\x+2*18,radius=4cm-0*0.2cm) circle (0.07cm);
}
\foreach \x in {0,5,10,15}
  \draw (canvas polar cs:angle=-90+18*\x,radius=4.7cm) node {(\x,0)};
\end{tikzpicture}
\end{center}
\caption{Example $h$-boundary of vertex $(16,0)$ with $ \nu=20$ and $h=5$.}
\label{fig:three}
\end{figure}

Here is the relevant fact about $h$-boundaries. If $B$ is an $h$-boundary of a vertex~$(x,0)$ and
$\sigma(B)$ has parity-$s$ ones then there is no contour of~$\sigma$ which contains an
edge whose dual connects two vertices in~$B$.

If $B$ is an $h$-boundary of vertex $(x,0)$
and $B'$ is an $h'$-boundary of $(x,0)$, then
we say that $B$ is inside of $B'$ if
every path in $C_\size$ from~$B'$ to $(x,0)$ passes through $B$.
Suppose that~$B$ and~$B'$ are $h$-boundaries of $(x,0)$
and, for some $s\in\{0,1\}$,
$\sigma(B)$ has parity-$s$ ones and $\sigma(B')$ has parity-$(1\oplus s)$ ones.
Then $B\cap B'=\emptyset$, so exactly one of the following occurs.
\begin{itemize}
\item $B$ is inside of $B'$, or
\item $B'$ is inside of $B$.
\end{itemize}

\begin{definition}\label{def:phase}
Suppose that $s\in \{0,1\}$ and that $k$ and $d$ are positive integers. Let~$\nu=2kd$. Suppose
that $\sigma$ is a configuration
$\sigma:\configs{V(C_\size)}$.
We say that $\sigma$ has 
phase~$s$ if the following holds for every terminal~$(x,0)$.
\begin{itemize}
\item $(x,0)$ has a $d$-boundary $B$ for which $\sigma(B)$ has parity-$s$ ones.
\item For every $d$-boundary $B'$ of $(x,0)$ for which $\sigma(B')$ has parity-$(1{\oplus} s)$ ones,
$B'$ is inside of $B$.
\end{itemize}
\end{definition} 

Note that a configuration~$\sigma$ can have exactly one phase (phase~$0$ or phase~$1$)
or it can have no phase.
Suppose that $\size \geq 1$ and that the configuration
$\sigma:\configs{V(C_\size)}$ has phase~$s$.
Say that a $d$-boundary $B$ of  a vertex $(x,0)$ is \emph{consistent}
if   $\sigma(B)$ has  parity-$s$ ones.
From the consistent $d$-boundaries, we want to select a canonical one, 
that is in some precise sense ``outermost'' and also ``minimal''.
For each  terminal~$(x,0)$, let
$\calBhat_x(\sigma)$ be the union of
all $d$-boundaries  of $(x,0)$ which are consistent.
Observe that $B=\calBhat_x(\sigma)$ satisfies the first three bullet points 
in Definition~\ref{def:boundaries}, but not the final one, as the subgraph $C_\size[B]$
of $C_\size$ induced by $B$ may not be connected.  Suppose that $C_\size[B]$ 
has $j$ connected components.  Partition 
$\calBhat_x(\sigma)=\calBhat^1_x(\sigma)\cup\cdots\cup\calBhat^j(\sigma)$ so that
$C_\size[\calBhat^1_x(\sigma)], \ldots, C_\size[\calBhat^j_x(\sigma)]$ is an enumeration 
of these $j$~connected components.  Each set $\calBhat^i_x(\sigma)$ is itself a $d$-boundary.
To see this, consider any vertex $v\in\calBhat^i_x(\sigma)$.  From the construction
of $\calBhat_x(\sigma)$, this vertex is contained in some $d$-boundary, which in turn is contained
in $\calBhat_x^i(\sigma)$.  So $\calBhat_x^i(\sigma)$ satisfies the first three bullet points
in Definition~\ref{def:boundaries}, in addition to inducing a connected graph. 
 
From the first bullet point of Definition~\ref{def:boundaries} it follows that 
$\calBhat^1_x(\sigma),\ldots,\calBhat^j_x(\sigma)$ are nested;  suppose the numbering
indicates the level of nesting, with $\calBhat^1_x(\sigma)$ being the outermost.
We now want to identify a minimal $d$-boundary within~$\calBhat^1_x(\sigma)$.  
Let 
$$\Ext\calBhat_x^1(\sigma)=\{v\in V(C_\size)\mid\text{there is a $*$-path 
from $v$ to $(x,\size)$ in $V(C_\size)\setminus\calBhat_x^1(\sigma)$}\}.$$
denote the set of vertices lying in the ``exterior'' of $\calBhat^1_x(\sigma)$.   Finally define
$$\calB_x(\sigma) = \{v \in \calBhat^1_x(\sigma)\mid 
\text{$v$ is $*$-adjacent to some vertex in $\Ext\calBhat^1_x(\sigma)$}\}.$$
Note that $\calB_x(\sigma)$ is a $d$-boundary of $(x,0)$  which is consistent.
(Informally, $\calB_x(\sigma)$ is the outermost such $d$-boundary.)
To see this, observe that any path in the graph~$C_\size$ from $(x,0)$ to
$(x,\size)$ has a last vertex in the set $\calBhat^1_x(\sigma)$, and this vertex must be in $\calB_x(\sigma)$;
this deals with the first bullet point in Definition~\ref{def:boundaries}.  
Consider the trail in the dual graph $C_\size^*$
separating $\calB_x(\sigma)$ and $\Ext\calBhat^1_x(\sigma)$;  the trail in the primal graph
that shadows it at $*$-distance~$\frac12$ inside takes in all the vertices of $\calB_x(\sigma)$
and establishes connectivity of $C_\size[\calB_x(\sigma)]$.  This deals with the final bullet
mark, and the remaining two are immediate.

The $d$-boundary $\calB_x(\sigma)$ is our desired canonical $d$-boundary, and it has the 
following important property.  If $\sigma'$ is a configuration that agrees with~$\sigma$
on $\calB_x(\sigma)\cup \Ext\calBhat_x^1(\sigma)$, then $\calB_x(\sigma')=\calB_x(\sigma)$.
The reason is as follows.  The set $B=\calB_x(\sigma)$ is a $d$-boundary of $(x,0)$ which is 
consistent with respect to~$\sigma'$, i.e., $\sigma'(B)$ has parity-$s$ ones.
It therefore gets incorporated into $\calBhat_x(\sigma')\supseteq\calB_x(\sigma)$
and hence into $\calBhat_x^1(\sigma')\supseteq\calB_x(\sigma)$.  So $\Ext\calB_x^1(\sigma')
=\Ext\calB_x^1(\sigma)$ and 
$\calB_x(\sigma')=\calB_x(\sigma)$.  
This fact becomes significant when we come to consider events supported on spins in the 
interior~$U=V(C_\size)\setminus(B\cup\Ext B)$ 
of some $d$-boundary~$B$.  Specifically, conditioning on the event 
$\calB_x(\bsigmasize)=B$ (and on the phase~$s$ of~$\bsigmasize$) is equivalent to  
selecting $\bsigma_\size(U)$ according to the Gibbs distribution,
with the boundary condition ``$\bsigma_\size(B)$ has parity-$s$ ones''.
We refer to this property as {\it canonicity\/} of~$\calB_x(\sigma)$.
  
Before proceeding we need some definitions and observations concerning connected 
subgraphs~$K$ of the graph $(V(C^*_\size),\sigma^*)$.
The \emph{$*$-diameter of $K$} is the maximum, over pairs of vertices in~$K$,
of the $*$-distance between those vertices.
We say that~$K$ \emph{reaches the lower boundary of~$C_\size$\/}
if it contains a vertex of the form $(x,-\frac12)$ for some~$x$.
We say that it reaches the upper boundary of~$C_\size$
if it  contains a vertex of the 
form  $(x,\size+\frac12)$. We say that
it is a \emph{cross subgraph} if it reaches both boundaries.
We say that   $K$ \emph{wraps around\/} if it contains the image of some path from
$(0,y)$ to $(2\size,y)$ in $\mathbb{Z}\times\{0,\dots,\size\}$, under
the quotient map to $V(C_\size)$.   
We say that~$K$ is {\it local\/} if 
it is not a cross subgraph and does not wrap around.
We say that $K$  \emph{intersects} 
$U_{x,h}$ if 
some edge $e^*$ in~$K$ is dual to an edge $e\in E(C_\size)$ 
with both endpoints in~$U_{x,h}$.  
A contour is a connected subgraph of $(V(C^*_\size),\sigma^*)$
so all of the above definitions apply to contours.
We refer to the connected components of $(V(C^*_\size),\sigma^*)$ as $\sigma^*$-components.

\begin{lemma}\label{lem:contourscomponents}
If $\sigma$ has only local contours then $\sigma^*$ has only local components.
\end{lemma}

\begin{proof}
Suppose to the contrary that $\sigma^*$ has non-local component~$K$.  Consider first 
the case of a cross component.

So suppose $K$ reaches both the upper and lower boundaries of $C_\size$, and that $K$ has
$2j>0$ degree-1 vertices.  (All vertices other than the degree-1 vertices have even degree,
and the number of odd-degree vertices in a graph is even.)
By adapting the standard algorithm for finding an Eulerian trail
in a (connected) Eulerian graph, we may decompose $K$ into $j$~contours beginning and 
ending at degree-1 vertices.  The method is as follows.
Starting at a degree-1 vertex, trace out a trail in~$K$
subject only to the rule that we must turn through a right-angle at any degree-4 vertex.
This  
trail
can only end at another degree-1 vertex.  The trail so formed is a contour;
remove the trail from~$K$ and repeat $j-1$ more times to obtain $j$ contours in total.
If any edges of~$K$ remain, start at any remaining vertex and trace out a closed 
trail that returns to the start vertex.  Again, the rule is always to turn through a right-angle
at any degree~4 vertex.  Repeat until there are no edges remaining in~$K$.  We are 
left with $j$ non-closed contours and an unspecified number of closed ones.  Whenever
a non-closed contour meets a closed one, we may splice the latter into the former,
reducing the number of closed contours by one.  Repeating as necessary, we obtain the 
sought-for decomposition of $K$ into $j$~contours beginning and ending at 
degree-1 vertices. 

If one of these $j$ contours joins the upper and lower boundaries of $C_\size$ we
are done, as we have already found a cross contour and obtained a contradiction.  
Otherwise, there must be at least one vertex at which a lower-to-lower contour 
touches a upper-to-upper contour.  Simply reroute the trails at this vertex to obtain 
two cross contours.  

Now consider the case where $K$ wraps around.  We may assume that $K$ does not reach
one of the boundaries of $C_\size$, say the upper one.  Trace a closed trail along the 
upper boundary of~$K$:  this trail is a contour that wraps around, providing a contradiction.
\end{proof}

Since non-local contours are unlikely, Lemma~\ref{lem:contourscomponents} allows
us to concentrate on local $\sigma^*$-components.  A $\sigma^*$-component $K$
that is local has a well defined inside and outside, and a boundary that is a
valid contour.  (More precisely, there is a 
canonical
contour that has exactly
the same edges as the boundary of~$K$.)
If $K$ reaches neither the upper nor lower boundary of $C_\size$,
then we may trace clockwise around $K$, always taking the leftmost option, until
we return to our starting point.  This procedure yields a simple contour;  we refer 
to vertices of $V(C_\size)$ that lie within this contour as forming the \emph{interior
of $K$}, denoted $\Int K$.

If $K$ reaches the lower boundary but not the top (or vice versa), 
then a slightly modified construction can be used.  First 
lift $K$ to a grid:  for sufficiently large $N$, there is a
connected subset $\widehat{K}$ of
$\{1,\dots,N\}\times\{0,\dots,\nu\}$ which maps bijectively to~$K$
under the quotient map to $C_\nu$.  Note that lifting can only increase the 
diameter of $\widehat K$ relative to $K$.  We now have a natural 
ordering of the degree-1 vertices of~$K$, namely by increasing 
$x$-coordinate.  Start at the least degree-1 vertex in this ordering and 
and trace the boundary of $\widehat K$ in a clockwise-leftmost fashion until the greatest
degree-1 vertex is reached.  This procedure yields a simple contour 
which partitions the vertices of $V(C_\size)$ into an inside and an outside (containing the
point $(0,\size)$);  again we refer to the former as the \emph{interior of $K$}.  

\begin{lemma}\label{lem:tempone}
Consider $\sigma\colon\configs{V(C_\size)}$.
Let $h \leq \nu$ be an integer multiple of~8.
Suppose 
\begin{itemize}
\item $\sigma$ contains only local contours.
\item $\sigma$ contains no contour of length at least~$h/8$ that intersects $U_{x,h/2}$.
\item $\sigma$ contains no simple contour separating 
$U_{x,h/4}$
from the upper boundary of $C_\size$.
\end{itemize}
Then every $\sigma^*$-component that intersects $U_{x,h/2}$ has 
$*$-diameter at most~$h/8$.
\end{lemma}
  
\begin{proof}
Suppose $K$ is a 
$\sigma^*$-component
intersecting $U_{x,h/2}$.  By Lemma~\ref{lem:contourscomponents},
$K$ is local.  The interior of $K$ contains some vertex in $U_{x,h/2}$.  By assumption, the 
contour defined by the boundary of~$K$ does not separate 
$U_{x,h/4}$ from the upper boundary of~$C_\size$, so it does not
separate 
$U_{x,h/2}$
from the upper boundary of~$C_\size$.
The only remaining possibility is that this contour intersects $U_{x,h/2}$, and hence has length 
at most $h/8$.  It follows that $K$ has $*$-diameter at most~$h/8$.
\end{proof}

\begin{lemma}\label{lem:august8}
Suppose that $h$ is a sufficiently large multiple of~$8$, and 
$\size \geq h$.
Consider a configuration $\sigma\colon\configs{V(C_\size)}$ and a terminal~$(x,0)$.
Suppose that the following are true.
\begin{itemize}
\item $\sigma$ has no cross contours.
\item $\sigma$ has no simple contour of length at least~$h/8$ that intersects~$U_{x,h/2}$.
\item $\sigma$ has no simple contour separating $U_{x,h/4}$ from the upper boundary of~$C_\size$.
\end{itemize}
Then, for some $s\in \{0,1\}$,
\begin{enumerate}
\item $(x,0)$ has an $h$-boundary~$B$ for which $\sigma(B)$ has parity-$s$ ones.
\item There is no $h$-boundary~$B'$ of $(x,0)$ for which $\sigma(B')$ has parity-$(1{\oplus} s)$ ones.
\item 
For any terminal $(x',0)$ that
has an $h$-boundary $B'$ for which $\sigma(B')$ has parity-$s'$ ones,
if $\sigma$ has
no contour separating $U_{x,h/4}$ from $U_{x',h/4}$  
then $s'=s$.
\item If $\sigma(B_{0,\size})$ has parity-$s'$ ones then $s'=s$.
\end{enumerate}
\end{lemma}

\begin{proof} 
There are no contours that wrap around, since any such contour would either intersect $U_{x,h/2}$,
or would separate $U_{x,h/4}$ from the upper boundary of $C_\size$.
Thus, by Lemma~\ref{lem:contourscomponents}, all $\sigma^*$-components are local.
Let $S$ be the set of all vertices in $V(C_\size)$ that are not in the interior
of some $\sigma^*$-component.  That is 
$$S=V(C_\size)\setminus \bigcup\big\{\Int K\mid\text{$K$ is a $\sigma^*$-component}\big\}.$$
Note that $\sigma(S)$ has parity-$s$ ones,
for some $s\in\{0,1\}$.
Define $\Sbar=V(C_\size)\setminus S$.   Note that $\sigma(\Sbar\,)$ in general  
has mixed parity; the salient feature is that $\sigma(S)$ has consistent parity.

We work first towards conclusion~(1) of the lemma.
By Lemma~\ref{lem:tempone}, every $\sigma^*$-component that intersects~$U_{x,h/2}$ 
has $*$-diameter at most~$h/8$.  
Now (informally) we will construct the required $h$-boundary by tracing 
round the inside of $B_{x,h/2}$,
making a detour towards $(x,0)$ around any $\sigma^*$-components that stand in the way.
(Recall that $B_{x,h/2}=U_{x,h/2}\setminus U_{x,h/2-1}$ is the ``goalpost'' at
distance $h/2$ from $(x,0)$.)
This strategy ensures we remain in the set~$S$ and, since all the $\sigma^*$-components
are small, our detours will not be too great.
 
More formally, let $W$ be the union of the set $V(C_\size)\setminus U_{x,h/2}$ together with 
any sets of the form~$\Int K$ that intersect $B_{x,h/2}$.  That is,
$$
W=(V(C_\size)\setminus U_{x,h/2})\cup \bigcup\{\Int K\mid\text{$K$ is a $\sigma^*$-component and $\Int(K)\cap
B_{x,h/2}$}\not=\emptyset\}.
$$
The set 
$$\partial W=\{v\in V(C_\size)\mid v\notin W\text{ and $v$ is $*$-adjacent to some vertex in $W$}\}$$
is almost the $h$-boundary $B$ that we seek.
Observe that any set of the form $\Int K$ is contained in
a maximal set of the form $\Int K'$, and the $*$-neighbours of $\Int K'$
are all in $S$.
Thus $\partial W$ is a subset of $S$, and
necessarily has parity-$s$ ones.

And as we shall see presently, $\partial W$ satisfies the 
first three conditions of an $h$-boundary $B$ ---
(i)~every path from $(x,0)$ to $(x,\size)$ intersects $B$,
(ii)~$B\cap U_{x,h/4}=\emptyset$, 
(iii)~$B\subseteq U_{x,h/2}$ ---
but not necessarily the final one, namely:
(iv)~the induced graph $C_\size[B]$ is connected.  (There may be islands of vertices of $\partial W$
lying outside the $h$-boundary we are trying to home in on.)
However, we can ensure (iv) by defining $B$ to be the subset of vertices in $\partial W$
that can be reached from $(x,0)$ by a path in $C_\size$ whose vertices all lie in 
$V(C_\size)\setminus W$.

For (i), observe that any path 
from $(x,0)$ to $(x,\size)$ has a first vertex~$w$ in~$W$.  
The vertex immediately preceding~$w$ is 
not in~$W$ but is adjacent to a vertex in~$W$, and hence in~$B$.
(ii)~follows from the fact that every vertex in  
$W \cap U_{x,h/2}$
is within $*$-distance~$h/8$ of
a vertex in $B_{x,h/2}$.  (iii) is immediate from the construction.  
To see (iv), denote by $W^\circ$ the set of all vertices in $V(C_\size)$ 
that can be reached from $(x,0)$ by a path whose vertices all lie in 
$V(C_\size)\setminus W$.  Note that $C_\size[W^\circ]$ is connected and that $B\subseteq W^\circ$.  
Let $\varrho^*$ be the set of edges in $C_\size^*$ separating $W^\circ$ and $W$;  thus, $e^*\in\varrho^*$
iff $e$ has one endpoint in $W$ and the other in $W^\circ$.  Since $C_\size[W^\circ]$ is connected, 
the edges in $\varrho^*$ form a trail in $C_\size^*$ starting and ending at vertices 
with $y$-coordinate~$-\frac12$.  Following this trail anticlockwise, 
vertices in~$W^\circ$ lie to the left and those in~$W$ to the right.
In fact, the vertices immediately to the left of the $\varrho^*$-trail (i.e., those at
$*$-distance $\frac12$ from it) are precisely the 
vertices forming~$B$:  they are all $*$-adjacent to some vertex in $W$, and no other vertices in
$W^\circ$ have this property.  Thus, any two vertices in~$B$ are connected by a path, which is 
obtained by shadowing $\varrho^*$ at $*$-distance~$\frac12$. 

For conclusion (2) of the lemma, 
observe that 
there
cannot be an $h$-boundary~$B'$ 
of $(x,0)$
such that $\sigma(B')$ has parity-$(s{\oplus}1)$ ones,
as such a $B'$ would have to exist entirely within~$\Sbar$, and all 
$*$-connected components of $\Sbar$ are small ($*$-diameter at most $h/8$).

As for conclusion~(3), it is impossible for $s'\not=s$.  Consider the 
connected component of~$C_\size[S]$ containing $B$.
If $s'\not=s$ then the
boundary of this component contains a contour separating $B$ and~$B'$, and hence $U_{x,h/4}$ and $U_{x',h/4}$.
If $\sigma(B_{0,\size})$ has parity-$s'$ ones for $s'\neq s$ then the boundary of the
connected component of~$C_\size[S]$ containing~$B$ is a simple contour
separating~$B$ from the upper boundary of~$C_\size$ hence,
separating $U_{x,h/4}$ from the upper boundary of~$C_\size$, establishing~(4).
\end{proof}

\begin{corollary}\label{lem:contours_vs_boundary}
Suppose  that $k\geq 1$ that $d$ is a sufficiently large multiple of~$8$ and that
$\size = 2 dk$.  Suppose
$\sigma\colon \configs{V(C_\size)}$ has
no contours of length at least $d/8$. Either $\sigma$ has phase $0$ or
$\sigma$ has phase $1$.
\end{corollary}

\begin{proof}
Since there are no contours of length at least $d/8$ of any kind, the 
premises of Lemma~\ref{lem:august8} are all satisfied for every terminal $(x,0)$
and every other terminal $(x',0)$.
\end{proof}

The following monotonicity property is useful for comparing different
contours and boundary conditions.

\begin{lemma}\label{lem:vertex_monotonicity}
Let $x\in\mathbb{Z}/2\nu\mathbb{Z}$, let $B$ be an $h$-boundary of
$(x,0)$ for some $h$, and let $B'$ be an $h'$-boundary of $(x,0)$ for
some $h'$ such that $B$ is inside of $B'$.
Then 
\begin{align*}
&\Pr(\mbox{$\bsigmasize$ has parity-$0$ ones at $(x,0)$}
\mid
\mbox{${\bsigmasize}(B)$ has  parity-$0$ ones}) \\ &\qquad\geq
\Pr(\mbox{$\bsigmasize$ has parity-$0$ ones at $(x,0)$}
\mid
\mbox{${\bsigmasize}(B')$ has  parity-$0$ ones}).
\end{align*}

\end{lemma}
\begin{proof}
For each $S\subseteq V(C_\size)$ let $\sigma^S\colon\configs{C_\size}$ denote
the configuration which has parity-$0$ ones exactly on $S$.
So $\sigma^S(x,y)=1$ if and only if one of these two conditions hold:
$(x,y)\in S$ and $x+y$ is even, or $(x,y)\notin S$ and $x+y$ is odd.

For all $X,Y\subseteq V(C_\size)$ we have
\[
w_{C_\size}(\sigma^X)
w_{C_\size}(\sigma^Y)
=
w_{C_\size}(\sigma^{X\cap Y})
w_{C_\size}(\sigma^{X\cup Y})(\beta\gamma)^k,
\]
where $k$ is the number of edges $uv\in E(G)$ such that
$\{(\sigma^X(u),\sigma^X(v)),(\sigma^Y(u),\sigma^Y(v))\}=\{(0,0),(1,1)\}$
(so either $u\in X$ and $v\notin X$ and $u\notin Y$ and $v\in Y$,
or $v\in X$ and $u\notin X$ and $v\notin Y$ and $u\in Y$).
In particular, 
\[ w_{C_\size}(\sigma^X)
w_{C_\size}(\sigma^Y)
\leq
w_{C_\size}(\sigma^{X\cap Y})
w_{C_\size}(\sigma^{X\cup Y}). \]

Let
\begin{align*}
\mathcal{X}&=\{S \mid \{(x,0)\}\cup B'\subseteq S\subseteq V(C_\size)\},\text{ and}\\
\mathcal{Y}&=\{S \mid B\cup B'\subseteq S\subseteq V(C_\size)\}.
\end{align*}
By the FKG inequality \cite{FKG71} we have
\[ \left(\sum_{S\in\mathcal{X}}w_{C_\size}(\sigma^S)\right)
   \left(\sum_{S\in\mathcal{Y}}w_{C_\size}(\sigma^S)\right)
\leq
\left(\sum_{S\in\mathcal{X}\wedge\mathcal{Y}}w_{C_\size}(\sigma^S)\right)
   \left(\sum_{S\in\mathcal{X}\vee\mathcal{Y}}w_{C_\size}(\sigma^S)\right),
\]
where $\mathcal{X}\wedge\mathcal{Y}$ is the family of sets
$X\subseteq V(C_\size)$ such that $(\{(x,0)\}\cup B')\cap (B\cup B')=B'\subseteq X$, and
$\mathcal{X}\vee\mathcal{Y}$ is the family of sets $X\subseteq V(C_\nu)$
such that $\{(x,0)\}\cup B \cup B' \subseteq X$.
Finally,
\begin{align*}
\Pr(\mbox{$\bsigmasize$ has parity-$0$ ones at $(x,0)$} \mid \mbox{${\bsigmasize}(B')$ has parity-$0$ ones})
&=
\frac
{\sum_{S\in\mathcal{X}}w_{C_\size}(\sigma^S)}
{\sum_{S\in\mathcal{X}\wedge\mathcal{Y}}w_{C_\size}(\sigma^S)},\\
\noalign{\noindent{and}}
\Pr(\mbox{$\bsigmasize$ has parity-$0$ ones at $(x,0)$}\mid \mbox{${\bsigmasize}(B)$ has parity-$0$ ones})
&=
\frac
{\sum_{S\in\mathcal{X}\vee\mathcal{Y}}w_{C_\size}(\sigma^S)}
{\sum_{S\in\mathcal{Y}}w_{C_\size}(\sigma^S)}.\qedhere
\end{align*}
\end{proof}
 
\begin{lemma}\label{lem:boundary_terminals}
There is a $c>1$ such that the following is true   for 
any $k\geq 1$, any $s\in\{0,1\}$, any sufficiently large~$d$
which is a multiple of~$16$,
 and any assignment 
$\{B_x\}$ of $d$-boundaries    
for each terminal $(x,0)$: 
\begin{itemize}
\item
For every  parity-$s$ terminal $(x,0)$, 
$$|\Pr(\bsigmakd(x,0)=1 \mid \mbox{$\bsigmakd$ has  phase~$s$ and
 $\calB_x(\bsigmakd) = B_x$}) - \peq| \leq c^{-d}.$$
 \item 
 For every  parity-$(1{\oplus} s)$ terminal $(x,0)$,
 $$|\Pr(\bsigmakd(x,0)=1 \mid \mbox{$\bsigmakd$ has  phase~$s$ and
 $ \calB_x(\bsigmakd) = B_x$}) - \pneq| \leq c^{-d}.$$
   \end{itemize}
\end{lemma}

\begin{proof}
By symmetry (rotating the gadget so that parity-$0$ vertices become parity-$1$ vertices and
vice-versa), 
it suffices to prove the inequalities for $s=0$. For any $m\geq 1$, define
\begin{align*}
 \peq(m) &=  
\Pr(\bsigmasub{m}(0,0)=1 \mid \mbox{$\bsigmasub{m}(B_{0,m})$ has  parity-$0$ ones}), \mbox{and}\\
\pneq(m) &=  
\Pr(\bsigmasub{m}(1,0)=1 \mid \mbox{$\bsigmasub{m}(B_{1,m})$ has  parity-$0$ ones}).
\end{align*}  
Now note that for any $\size\geq m$, 
$\peq(m)$ and $\pneq(m)$ (as defined above) are
the same as the equivalent expressions in the gadget~$C_\size$. In particular,

\begin{align*}
 \peq(m) &=  
\Pr(\bsigmasize(0,0)=1 \mid \mbox{$\bsigmasize(B_{0,m})$ has  parity-$0$ ones}), \mbox{and}\\
\pneq(m) &=  
\Pr(\bsigmasize(1,0)=1 \mid \mbox{$\bsigmasize(B_{1,m})$ has  parity-$0$ ones}). 
\end{align*}

 Thus, by fixing large~$\nu$ and increasing~$m$, 
Lemma \ref{lem:vertex_monotonicity} implies that 
$\peq(m)$   is
weakly decreasing in $m$ and that  $\pneq(m)$ is weakly
increasing. 
Thus, $\peq = \lim_{m\to\infty} \peq(m)$ and
$\pneq = \lim_{m\to\infty} \pneq(m)$.
Also, for a parity-$0$ terminal $(x,0)$,
the target probability 
$$\Pr(\bsigmakd(x,0)=1 \mid \mbox{$\bsigmakd$ has   phase~$0$ and
 $  \calB_x(\bsigmakd) = B_x$})$$
 is between $\peq(d/2)$ and $\peq(d/4)$.
 Similarly, for a parity-$1$ terminal $(x,0)$,
 the target probability 
$$\Pr(\bsigmakd(x,0)=1 \mid \mbox{$\bsigmakd$ has   phase~$0$ and
 $  \calB_x(\bsigmakd) = B_x$})  $$  is between $\pneq(d/4)$ and $\pneq(d/2)$.
(Here we use crucially the canonicity property of 
$\calB_x(\cdot)$;  
refer to the discussion
following Definition~\ref{def:phase}.)
Thus it suffices
to show 
$$
\peq(d/4) \leq \peq +  c^{-d}\quad\mbox{and}\quad
\pneq(d/4)\geq \pneq - c^{-d}, 
$$

First we take a qualitative step.
Pick $w\geq 8d$  sufficiently large that  $\peq(w) \leq \peq + |U_{0,d/4}| (c')^{-d/16} $
and $\pneq(w) \geq \pneq - |U_{1,d/4}| (c')^{-d/16}$,
where $c'$ is the maximum of the constants given in Lemma~\ref{lem:longsimple_unlikelycond}
and Lemma~\ref{lem:newcontour}.
This can be done since 
$d$ is sufficiently large and
$\peq = \lim_{m\to\infty} \peq(m)$, 
and $\pneq = \lim_{m\to\infty}\pneq(m)$
though $w$ may be
quite a lot larger than~$d$.

For $i\in\{0,1\}$, let~$F_i$
be the event that
there is a $d/2$-boundary $B$ of vertex $(i,0)$ in gadget $C_w$
such that $\bsigmasub{w}(B)$ has  parity-$0$ ones. 
Recall from the definition that a $d/2  $-boundary of~$(i,0)$ is a subset
of~$U_{i,d/4}$.
Let~$E_i$ be the event
that $\bsigmasub{w}(B_{i,w})$ has  parity-$0$ ones.

For each $i\in\{0,1\}$,
applying 
Lemma~\ref{lem:longcross_unlikelycond},
and
Lemma~\ref{lem:longsimple_unlikelycond}
with $h=d/16$ and $U=U_{i,d/4}$, 
and
Lemma~\ref{lem:newcontour} 
with $h=d/8$, 
we find that, the conditional probability that the following hold, conditioned on~$E_i$,
is at least $1- 2 |U_{i,d/4}| (c')^{-d/16}$.
\begin{itemize}
\item $\bsigmasub{w}$ has no cross contour.
\item $\bsigmasub{w}$ has no simple contour of length at least $d/16$ 
which contains an edge between two vertices in~$U_{i,d/4}$.
\item $\bsigmasub{w}$ has no simple contour 
that separates $U_{i,d/8}$ from the upper boundary of~$C_\size$. 
\end{itemize}
Now, applying Lemma~\ref{lem:august8} with $h=d/2$ and $\size=w$ and $x=i$,  
if all of these hold and event $E_i$ occurs  then
event $F_i$ occurs. Thus,
$$\Pr(F_i \mid E_i) \geq 1- 2 | U_{i,d/4}| (c')^{-d/16}.$$
But by Lemma \ref{lem:vertex_monotonicity}, we have
\begin{align*}
\Pr(\bsigmasub{w}(0,0)&=1\mid F_0\wedge E_0 )\geq \peq(d/4), \mbox{ and}\\
\Pr(\bsigmasub{w}(1,0)&=1\mid F_1\wedge E_1 )\leq \pneq(d/4).
\end{align*}
So
\begin{align*} 
\peq(d/4) &\leq 
\frac{\Pr(  \bsigmasub{w}(0,0)=1 \wedge F_0 \mid E_0)}
{\Pr(F_0 \mid E_0)}
\leq 
\frac{\Pr( \bsigmasub{w}(0,0)=1 \mid E_0) }
{\Pr(F_0\mid E_0)}
= \frac{\peq(w)}{\Pr(F_0\mid E_0)}
\\
&\leq \frac{\peq(w)}{1- 2|U_{0,d/4}| (c')^{-d/16}}
\leq \peq(w) + 4 | U_{0,d/4}| (c')^{-d/16},
\end{align*}
since $2|U_{0,d/4}| (c')^{-d/16} \leq 1/2$.

A similar inequality holds for $\pneq(d/4)$:
\begin{align*} 
\pneq(d/4) &\geq 
\frac{\Pr(  \bsigmasub{w}(1,0)=1 \wedge F_1 \mid E_1)}
{\Pr(F_1 \mid E_1)}\\
&\geq 
\Pr( \bsigmasub{w}(1,0)=1 \mid E_1)-\Pr(\neg F_1\mid E_1)
\\
&\geq \pneq(w) - 2|U_{1,d/4}| (c')^{-d/16}.
\end{align*}
Thus,
\begin{align*}
\peq(d/4) &\leq \peq(w)+ 4 |U_{0,d/4}| (c')^{-d/16} \leq \peq+ 5 |U_{0,d/4}| (c')^{-d/16}, \mbox{ and}\\
\pneq(d/4) &\geq \pneq(w)
 -2 |U_{1,d/4}| (c')^{-d/16}  \geq \pneq
 -  3 |U_{1,d/4}| (c')^{-d/16}.  
\end{align*}
The result follows by noting that $|U_{i,d/4}|$ is $O(d^2)$ and
picking~$c=(c')^{1/17}$, say.
\end{proof}

We now prove the main proposition.

\begin{mainprop}
There is a $c>1$ such that,
if $d$ is a sufficiently large multiple of~$16$,
$k$ is an integer greater than or equal to~$1$ and
$\tau$ is a configuration $\tau:\configs{\terminals}$, then
$$|\Pr(\sigmarv(\terminals)=\tau)-\muall(\tau)| \leq c^{-d} k^2.$$
\end{mainprop}

\begin{proof}
Fix $k\geq 1$, $d$ 
a sufficiently large multiple of~$16$, 
and $\tau:\configs{\terminals}$.  
Let $c'$ be the 
minimum
value of the constant~$c$
from the Lemmas~\ref{lem:longsimple_unlikely},  \ref{lem:longcross_unlikely},
and~\ref{lem:boundary_terminals}.
  
The probability $\Pr(\sigmarv(\terminals)=\tau)$
is the sum of the following probabilities (conditioned on disjoint events)
\begin{itemize}
\item $\Pr(\sigmarv(\terminals)=\tau \mid 
\mbox{$\sigmarv$ does not have a phase}
) \Pr(\mbox{$\sigmarv$ does not have a phase})$;
\item (summed over all assignments $B_x$ of $d$-boundaries for each terminal $(x,0)$)
\begin{align*}
&\Pr(\sigmarv(\terminals)=\tau \mid
\mbox{ $\bsigmakd$ has   phase~$0$ and for all terminals $(x,0)$,
 $\calB_x(\bsigmakd) = B_x$}) \times \null\\
 &\quad \Pr(\mbox{ $\bsigmakd$ has   phase~$0$ and for all terminals $(x,0)$,
 $\calB_x(\bsigmakd) = B_x$}); \text{ and}
 \end{align*}
\item (summed over all assignments $B_x$ of $d$-boundaries for each terminal $(x,0)$)
\begin{align*}
&\Pr(\sigmarv(\terminals)=\tau \mid
\mbox{ $\bsigmakd$ has  phase~$1$ and for all terminals $(x,0)$,
 $\calB_x(\bsigmakd) = B_x$}) \times \null\\
 &\quad \Pr(\mbox{ $\bsigmakd$ has   phase~$1$ and for all terminals $(x,0)$,
 $\calB_x(\bsigmakd) = B_x$}).
 \end{align*}
\end{itemize}

By Lemmas~\ref{lem:longsimple_unlikely} and  \ref{lem:longcross_unlikely} and
Corollary~\ref{lem:contours_vs_boundary},
 the probability of the first of these
 is at most 
 $2 |V(\gadget^*)|(c')^{-d/8}$. 
 (We will use this below.)
 
Now consider an assignment $B_x$ of $d$-boundaries for each terminal $(x,0)$.
For any two terminals $(x',0)$ and $(x'',0)$,
the random variables
$\bsigmakd(x',0)$ 
are $\bsigmakd(x'',0)$ are independent, conditioned
on the fact that $\bsigmakd$ has a given phase, and
 for all terminals $(x,0)$,
 $\calB_{x}(\bsigmakd) = B_{x}$. 
 Also, by Lemma~\ref{lem:boundary_terminals}, for all $s\in\{0,1\}$,
 \begin{itemize}
\item
For every  parity-$s$ terminal $(x',0)$, 
\begin{align*}
\Bigl|\Pr\big(\bsigmakd(x',0)=1 &\bigm| \text{$\bsigmakd$ has   phase~$s$, and}\\ 
&\qquad\text{for all terminals $(x,0)$,
 $ \calB_x(\bsigmakd) = B_x$}\big) - \peq\Bigr| \leq (c')^{-d}.
\end{align*} 
 \item 
 For every  parity-$(1{\oplus} s)$ terminal $(x',0)$,
\begin{align*}\Bigl|\Pr\big(\bsigmakd(x',0)=1 &\bigm| \text{$\bsigmakd$ has   phase~$s$, and}\\
&\qquad \text{for all terminals $(x,0)$,
 $ \calB_x(\bsigmakd) = B_x$}\big) - \pneq\Bigr| \leq (c')^{-d}.
\end{align*} 
  \end{itemize}
 Now, for any probabilities
 $a_1,b_1,\dots,a_k,b_k$, we have
 \[\left|\prod_{i=1}^k a_i-\prod_{i=1}^k b_i\right|=\left|\sum_{j=1}^k a_1\dots a_{j-1}(a_j-b_j)b_{j+1}\dots b_k\right|\leq\sum_{i=1}^k \left|a_i-b_i\right|, \]
so if we fix a given phase, and
$\tau$ assigns spin~$1$ to~$k'$ terminals whose parity agrees with that
phase, and spin~$1$ to~$k''$ terminals whose parity disagrees with that phase
then, letting
$$\hat p = {(\peq)}^{k'} {(1-\peq)}^{k-k'} {(\pneq)}^{k''} {(1-\pneq)}^{k-k''}
,$$
we have
\begin{align*}
 & \hat p - 2 k (c')^{-d} \\
 &\quad\leq
 \Pr(\sigmarv(\terminals)=\tau \mid
\mbox{ $\bsigmakd$ has the given phase and for all terminals $(x,0)$,
 $\calB_x(\bsigmakd) = B_x$})\\&\quad
 \leq  \hat p + 2k (c')^{-d}.
 \end{align*}
 
 So summing up, 
 $\Pr(\sigmarv(\terminals)=\tau \mid
\mbox{ $\bsigmakd$ has the given phase})$ is between
 $  \hat p - 2 k (c')^{-d}$ 
 and
 $ \hat p + 2k (c')^{-d}$
 so, since the phases are equally likely,
 $$ \mu_{k,d}(\tau)-2k (c')^{-d} \leq
 \Pr(\sigmarv(\terminals)=\tau
 \mid \mbox{ $\bsigmakd$ has a phase}) \leq \mu_{k,d}(\tau)+ 2k (c')^{-d}.$$
 
Finally, since the probability that $\sigmarv$ has no phase
is at most 
$2 |V(\gadget^*)|(c')^{-d/8}$, 
as we observed above,
  $$ \mu_{k,d}(\tau)-2(k+|V(\gadget^*)|) (c')^{-d/8} \leq
 \Pr(\sigmarv(\terminals)=\tau) \leq \mu_{k,d}(\tau)+ 2(k+|V(\gadget^*)|) (c')^{-d/8}.$$ 
The proposition follows by choosing~$c$ to be sufficiently small with respect to~$c'$.
\end{proof}

Ideally, one would like to extend Proposition~\ref{prop:gadget} to a range of
parameter values that reach right to the edge of the non-uniqueness region
(i.e., the region in which the model on the infinite 2-dimensional square
grid has multiple Gibbs measures).  Unfortunately, it is not known rigorously 
whether the hard-core (independent set) model exhibits a sharp phase transition, 
let alone where that transition is located.  So a sharp version of 
Proposition~\ref{prop:gadget} would be a tougher proposition.
   
   \section{ Proof of Theorem~\ref{thm:main}} 
   \label{sec:proofmain}
  
  \subsection{Efficiently approximable reals}

\begin{lemma}
\label{lem:approxpeq} 
Suppose that $\beta$, $\gamma$ and $\lambda$ are efficiently approximable
reals satisfying~\eqref{eq:spincond}. 
Then $\peq$ and  $\pneq$  are efficiently approximable reals. \end{lemma}
\begin{proof} 
Recall that $\pneq>0$ (Lemma \ref{lem:little}).
Let $q$ be a 
multiple of~$16$
greater than $(2+\log_2(1/\pneq))/\log(c)$ where
$c$ is the constant given by Lemma \ref{lem:boundary_terminals}.
Consider the following algorithm.
\begin{itemize}
\item Input an error parameter $0<\epsilon<1/2$.
\item Set $m=q\lceil \log(\epsilon^{-1})\rceil$.

\item Compute rational approximations $\hat{\beta},\hat{\gamma},\hat{\lambda}$ satisfying
\begin{align*}
\beta e^{-\epsilon/8|E(C_m)|}&\leq\hat{\beta}\leq \beta e^{\epsilon/8|E(C_m)|},\\
\gamma e^{-\epsilon/8|E(C_m)|}&\leq\hat{\gamma}\leq \gamma e^{\epsilon/8|E(C_m)|},\text{ and}\\
\lambda e^{-\epsilon/8|V(C_m)|}&\leq\hat{\lambda}\leq \lambda e^{\epsilon/8|V(C_m)|}.
\end{align*}

\item Using the algorithm of \cite[Theorem 2.8]{YZ12}, compute
\begin{align*}
Z &= \sum_{\sigma}\hat{\beta}^{b(\sigma)}\hat{\gamma}^{c(\sigma)}\hat{\lambda}^{\ell(\sigma)},\\
Z' &= \sum_{\sigma: \sigma(0,0)=1}\hat{\beta}^{b(\sigma)}\hat{\gamma}^{c(\sigma)}\hat{\lambda}^{\ell(\sigma)},\text{ and}\\
Z'' &= \sum_{\sigma: \sigma(1,0)=1}\hat{\beta}^{b(\sigma)}\hat{\gamma}^{c(\sigma)}\hat{\lambda}^{\ell(\sigma)},
\end{align*}
where the sums range over configurations $\sigma$ of $C_m$
such that $\sigma(B_{0,m})$ has parity-$0$ ones.
\item Output $Z'/Z$ as the approximation to $\peq$, and $Z''/Z$ as the
 approximation to $\pneq$.
\end{itemize}
For the computation of $Z$, $Z'$, and $Z''$ we use the fact that the grid
graph $C_m\setminus B_{0,m}$ has treewidth~$m$  \cite[Corollary 89]{Bodlaender}.
We also use the fact that  its tree decomposition is easy  to compute.
So
this algorithm runs in time bounded by a polynomial in $1/\epsilon$.
We will show that the algorithm is an FPRAS for $\peq$ and $\pneq$.
Define
\begin{align*}
W &= \sum_{\sigma}\beta^{b(\sigma)}\gamma^{c(\sigma)}\lambda^{\ell(\sigma)},\\
W' &= \sum_{\sigma: \sigma(0,0)=1}\beta^{b(\sigma)}\gamma^{c(\sigma)}\lambda^{\ell(\sigma)},\text{ and}\\
W'' &= \sum_{\sigma: \sigma(1,0)=1}\beta^{b(\sigma)}\gamma^{c(\sigma)}\lambda^{\ell(\sigma)},
\end{align*}
where the sums range over configurations $\sigma$ of $C_m$
such that $\sigma(B_{0,m})$ has parity-$0$ ones.

For any $\sigma$ we have
\[\beta^{b(\sigma)}\gamma^{c(\sigma)}\lambda^{\ell(\sigma)}e^{-\epsilon/4}
\leq
\hat{\beta}^{b(\sigma)}\hat{\gamma}^{c(\sigma)}\hat{\lambda}^{\ell(\sigma)}
\leq
\beta^{b(\sigma)}\gamma^{c(\sigma)}\lambda^{\ell(\sigma)}e^{\epsilon/4}.
\]
This implies $e^{-\epsilon/4}W\leq Z\leq e^{\epsilon/4}W$ and
similarly for $Z'$ and $Z''$, and therefore $e^{-\epsilon/2}W'/W\leq
Z'/Z\leq e^{\epsilon/2}W'/W$ and $e^{-\epsilon/2}W''/W\leq Z''/Z\leq
e^{\epsilon/2}W''/W$.  We will show
\begin{align}
\peq &\leq W'/W \leq \peq e^{\epsilon/2},\text{ and} \label{eq:peq_approx}\\
e^{-\epsilon/2}\pneq&\leq W''/W \leq \pneq.\label{eq:pneq_approx}
\end{align}
The quotients $W'/W$ and $W''/W$ are just the probabilities that an even or odd
terminal gets assigned $1$ in a random configuration of $C_m$,
conditioned on a certain $2m$-boundary.  By Lemma
\ref{lem:vertex_monotonicity} we have $\peq\leq W'/W$ and $W''/W\leq
\pneq$ for any $m$, establishing the first inequality in
\eqref{eq:peq_approx} and the second inequality in
\eqref{eq:pneq_approx}.

By  
Lemma \ref{lem:boundary_terminals}, there
exists $c>1$ such that  
$$ 
W'/W\leq \peq+c^{-q \log(\epsilon^{-1})}
= \peq(1+\epsilon^{q\log(c)}/\peq).$$
Since
$$\epsilon^{(q \log(c)-1)}
\leq (1/2)^{(q \log(c)-1)}
\leq \pneq/2,$$
which is  less than $\peq$ by Lemma~\ref{lem:peqgreater},
we have
$$W'/W \leq \peq (1+\epsilon) \leq e^\epsilon.$$
This establishes \eqref{eq:peq_approx}.  Similarly, by Lemma
\ref{lem:boundary_terminals},
\begin{align*}
W''/W&\geq \pneq-c^{-q \log(\epsilon^{-1})}\\
&\geq \pneq(1-\epsilon^{q\log(c)}/\pneq)\\
&\geq \pneq(1-\epsilon/2)\\
&\geq \pneq e^{-\epsilon}.
\end{align*}
This establishes \eqref{eq:pneq_approx}.
\end{proof}
   
Lemma~\ref{lem:approxpeq} gives us a way to obtain multiplicative approximations
$\hatpeq$ and $\hatpneq$ of the real numbers~$\peq$ and~$\pneq$.
When we use these approximations, we will need to know
that $1-\hatpeq$ and $1-\hatpneq$ are also good multiplicative approximations
to~$1-\peq$ and~$1-\pneq$, respectively. As we show below, this follows from the fact that $\peq$ and $\pneq$ are in $(0,1)$ (which follows from Lemma~\ref{lem:little} and Lemma~\ref{lem:peqgreater}).
The following lemma gives us what we need. 
The reason for introducing the rational~$p'$ in the statement of the lemma is
that, since it is rational, it can be hard-wired into any algorithms (whereas a real number can't be).
\begin{lemma}
\label{lem:technical}
Suppose that $p\in(0,1)$ is an efficiently approximable real number.
Let $p'$ be a positive rational with $p<p'<1$.
For any $\delta\in(0,1)$,
and any real number $\hatp$ satisfying 
$e^{- \delta(1-p')/2} \hatp \leq p \leq e^{\delta(1-p')/2} \hatp$, we have
$e^{- \delta} (1-p) \leq 1-\hatp \leq e^\delta (1- p)$. 
\end{lemma}
\begin{proof}
Let $\delta' = \delta(1-p')/2$.
Since $\hatp \geq e^{-\delta'} p \geq p(1-\delta') \geq p-\delta'$
and similarly $p\geq \hatp - \delta'$,
we have
$$(1-p)\left(1-\frac{\delta (1-p')}{2(1-p)}\right)=
(1-p)-\delta' \leq 1-\hatp \leq (1-p) + \delta'=(1-p)\left(1+\frac{\delta(1-p')}{2(1-p)}\right).
$$
Thus,
$$(1-p)(1- \delta/2) \leq 1-\hatp \leq (1-p)(1+\delta/2),$$ which suffices.
\end{proof}

The following problem is NP-complete~\cite{gjs}.
\prob{\PlanarCubicIS.}
{A   planar cubic graph $G$ and a positive integer~$h$.}
{``Yes'', if $G$ contains an independent set of size $h$, and ``No'', otherwise.}
Suppose that $\beta$, $\gamma$ and $\lambda$ are efficiently 
approximable reals satisfying \eqref{eq:spincond}. 
We will give a randomised polynomial-time algorithm for \PlanarCubicIS,
using as an oracle, an FPRAS  for $\PlanarTwoSpin(\beta,\gamma,\lambda)$.
The oracle will be used to  
approximate $Z_{1,\tilde\gamma,\tilde\lambda}(G)$,
for some suitably-defined parameters $\tilde\gamma$ and $\tilde\lambda$,
where $\tilde\gamma$ is exponentially small in $|V(G)|$ and $\tilde\lambda$ is
exponentially large.
From this, it will be easy to determine whether~$G$ has an independent set of size~$h$.   
\begin{lemma}
\label{lem:gadget}
Suppose that $\beta$, $\gamma$ and $\lambda$ are efficiently approximable reals satisfying~\eqref{eq:spincond}.
There is a polynomial-time randomised  algorithm that, given a planar cubic graph~$G$ with
$|V(G)|$ sufficiently large, 
outputs 
planar graphs~$J$ and~$J'$ with maximum degree at most~$4$
and randomised approximation schemes for 
positive reals $K$, $\tilde\gamma$ and $\tilde\lambda$.
The running time of each of these approximation schemes  is bounded
from above by a polynomial in~$|V(G)|$ and the desired accuracy parameter~$\epsilon$.
With probability at least $14/15$,  
the parameters satisfy
$\tilde\lambda \geq 4^{|V(G)|}$ and $\tilde\gamma\leq {\tilde\lambda}^{-|V(G)|}$
and
\begin{equation}
\label{eq:lem15}
e^{-1/4} Z_{1,\tilde\gamma,\tilde\lambda}(G)
\leq K \frac{Z_{\beta,\gamma,\lambda}(J')}{Z_{\beta,\gamma,\lambda}(J)}
\leq e^{1/4} Z_{1,\tilde\gamma,\tilde\lambda}(G).
\end{equation} 
\end{lemma}

\begin{proof}
 
Let $G=(V,E)$ be a planar cubic graph  and let $n$ denote $|V|$.

The algorithm for constructing~$J$ and~$J'$ uses a quantity $\delta\in(0,1)$.
It will be important for the proof that $\delta$ is sufficiently small.
Rather than giving a technical definition here, we introduce 
upper bounds on~$\delta$ in natural places throughout the proof. The reader
can verify that each of these upper bounds is at least the inverse of a polynomial in~$n$
(so the algorithm runs in polynomial time).
 
The first step is to use the given FPRASes for $\beta$, $\gamma$ and $\lambda$, and the FPRASes
for $\peq$ and $\pneq$ from Lemma~\ref{lem:approxpeq} to compute
values $\hatbeta$, $\hatgamma$, $\hatlambda$, $\hatpeq$ and $\hatpneq$ satisfying 
\begin{align}
\begin{split}
\label{goodapprox}
e^{-\delta/3} \beta &\leq \hatbeta \leq e^{\delta/3} \beta,\\
e^{-\delta/3} \gamma &\leq \hatgamma \leq e^{\delta/3} \gamma,\\
e^{-\delta/3} \lambda &\leq \hatlambda \leq e^{\delta/3} \lambda,\\
e^{-\delta/3} \peq &\leq \hatpeq \leq e^{\delta/3} \peq,\\
e^{-\delta/3} \pneq &\leq \hatpneq \leq e^{\delta/3} \pneq,\\
e^{-\delta/3} (1-\peq) &\leq  1-\hatpeq \leq e^{\delta/3} (1-\peq),\\
e^{-\delta/3} (1-\pneq) &\leq 1-\hatpneq \leq e^{\delta/3} (1-\pneq),\text{ and}\\
\hatbeta &\geq 1.\\ 
\end{split}
\end{align}
The first five  lines   in~\eqref{goodapprox} follow directly from the definition of FPRAS
in Section~\ref{sec:FPRAS}.  
The next two  lines follow from Lemma~\ref{lem:technical}, using the fact
that $\peq$ and $\pneq$ are in $(0,1)$, as argued just before Lemma~\ref{lem:technical}.
Since $\beta\geq 1$ by \eqref{eq:spincond},
we can ensure that $\hatbeta\geq 1$ 
by taking $\hatbeta$ to be the maximum of~$1$ and the output of the FPRAS.
For this step we adjust the failure probability of the FPRASes (as described in Section~\ref{sec:FPRAS}) so that the probability that  
Equation~\eqref{goodapprox} fails to hold is at most~$1/15$.  
Note that the running time of the FPRASes is polynomial in $1/\delta$ (even though the application
of Lemma~\ref{lem:technical} means that we have to call the FPRASes for $\peq$ and $\pneq$
with slightly smaller values~$\delta'$.).

We will show below how to 
use $G$ and these approximations to
define positive integers $k_1$, $k_2$ and $d$, which will be used in the construction of~$J$
and~$J'$.
These quantities will be bounded from above by a polynomial in~$n$.

We first show how to construct~$J$ and~$J'$, using $k_1$, $k_2$, $d$
and  
$k = \max(k_1, 3 k_2)$.  
The high-level construction is illustrated in Figure~\ref{fig:J}.
The idea is that $J$ encodes just the vertices of $G$, while $J'$
extends $J$ to account the edges of~$G$, and includes additional
structures called ``bristles''.  Note that Figure~\ref{fig:J} represents just
a fragment of~$G$, and is schematic only:  parity-0 and parity-1 
terminals are in reality interleaved, but here they are separated, for clarity. 

\begin{figure}[!h]
\begin{center}
\begin{tikzpicture}[scale=0.9]
\newcommand\gjexamplelines{(1,3) -- (0,4) (2,1) -- (2,-0.3) (3,2) -- (4.3,2) (1,3) -- (2,1) -- (3,2) -- (1,3)};
\newcommand\bristles{(1,3) -- (1.5,3.5) (2,1) -- (1.3,1) (3,2) -- (3,2.7)};  
\draw[thick] \gjexamplelines;
\foreach \x/\y in {1/3,2/1,3/2} \fill (\x,\y) circle (0.1cm);
\draw[double, thick, ->] (5,2) -- (6,2);
\draw[xshift=11cm, line width=0.2cm] \gjexamplelines \bristles;
\draw[white, xshift=11cm, line width=0.1cm] \gjexamplelines \bristles;
\draw[xshift=11cm, line width=0.05cm] \gjexamplelines \bristles;
\foreach \x/\y in {1/3,2/1,3/2} {
  \draw[fill=gray] (\x+6,\y) circle (0.4cm);
  \draw[fill=white] (\x+6,\y) circle (0.2cm);
  \draw[fill=gray] (\x+11,\y) circle (0.4cm);
  \draw[fill=white] (\x+11,\y) circle (0.2cm);
}
\draw (10.5,2) node {and};
\draw (2,-1) node {$G$}; \draw (8,-1) node {$J$};  \draw (13,-1) node {$J'$};
\end{tikzpicture}
\end{center}
\caption{An illustration of how $G$ is transformed into the graphs $J$
  and $J'$.  A fragment of $G$ is shown on the left. 
  The graph $J$ is a collection of copies of $\gadget$, one for each vertex of $G$.
  A fragment of $J$ is shown in the middle.
  The copies of $\gadget$ are
  shown as grey annuli.
  The corresponding fragment of $J'$ is shown on the right.
  The long stripes represent the
  sets of edges between copies of $\gadget$ in $J'$,  
  while the short stripes are edges to ``bristles'', described elsewhere.}
  \label{fig:J}
\end{figure}
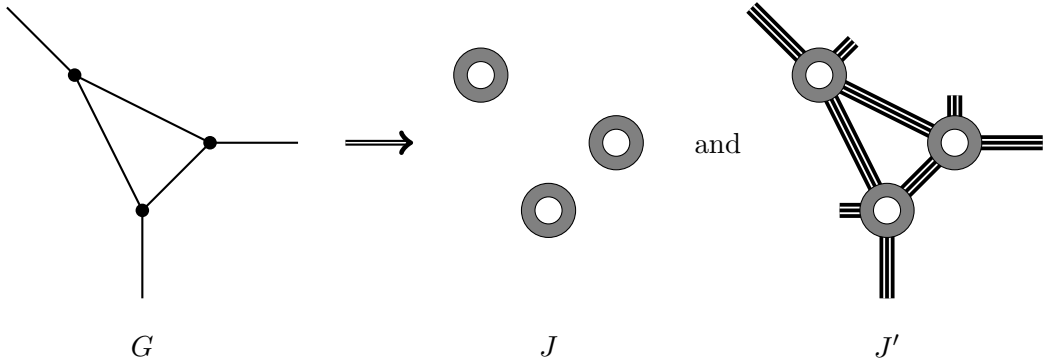 

The construction of~$J$ is straightforward. 
Essentially, $J$ consists of $|V|$ copies of $\gadget$, with one copy
for every vertex in $V$. Thus,
the vertex set 
$V(J)$ is the set of ordered pairs $V(J) = V \times V(\gadget)$ and
the edge set $E(J)$ is given by $E(J)=V \times E(\gadget)$.
We will use $C[u]$ to denote the gadget corresponding to vertex $u\in V$.
Formally, $C[u]$ is the graph with vertex set 
$\{u\} \times V(\gadget)$ and edge set $\{u\} \times E(\gadget)$.
To simplify the notation, 
for $u\in V$ and $0\leq j \leq k-1$, 
let $\oddterminal[u,j]$ denote the $j$'th
 parity-$1$ terminal of~$C[u]$.
Formally, this is the vertex $(u,(4 j d + 1, 0))$ of~$J$.
Similarly, let $\eventerminal[u,j]$ denote the $j$'th  parity-$0$
terminal of~$C[u]$. Formally, this is the vertex $(u,(4 j d + 2 d, 0))$ of~$J$. Let~$T[u]$ be
the set of terminals of~$C[u]$. Let $\mu^0_u$, $\mu^1_u$ and $\mu_u$ be the distributions on
configurations $\sigma:\configs{T[u]}$ corresponding to the distributions
$\muall^0$, $\muall^1$ and $\muall$  
defined in Section~\ref{sec:terminals}.
 
We now define~$J'$.  Informally, we extend $J$ by adding, for each edge $(u,v)$ of~$G$, a
$k_2$-edge matching connecting parity-0 terminals in $C[u]$ and $C[v]$.
In addition, for each vertex~$v$ in~$G$ we add a $k_1$-edge matching from
parity-1 terminals in $C[v]$ to $k_1$ new vertices (the bristles mentioned earlier).
The rough idea behind the construction is as follows.  Each gadget $C[v]$ is in one 
of two phases, 0 or 1.  The matchings between gadgets are there to discourage adjacent 
gadgets from being both in phase~1.  The bristles are there to encourage each gadget
to be in phase~1.  By playing off these two tendencies (by adjusting $k_1$ and~$k_2$)
we effectively pick out maximum independent sets in~$G$.

To simplify the formal description of~$J'$, 
consider a planar embedding of~$G$, and divide each edge into two ``half-edges''.
Label the half edges incident at vertex~$u$ of~$G$ by
$u_0$, $u_1$ and $u_2$ in clockwise order in the plane (the choice of
starting half-edge $u_0$ being arbitrary).
The labelling just applied defines a matching $\matching$
on $\bigcup_{u\in V} \{u_0,u_1,u_2\}$ in the natural way:
for all $u,v\in V$ and $i,j\in\{0,1,2\}$, we match $u_i$ and 
$v_j$ (i.e., include $(u_i,v_j)$ in~$\matching$) iff $u_i$ and~$v_j$ 
are half-edges constituting a single edge of~$G$.
The vertex set $V(J')$ consists of $V(J)$, together with  
a set of $n k_1$ new vertices, called ``bristles''.
Formally, 
$V(J') = V(J) \cup \{ (u,j) \mid u\in V, 0\leq j \leq k_1-1\}$.
Finally, the edge set of $J'$ consists of $E(J)$, together with
new edges connecting the bristles to the  parity-$1$ terminals
of the gadgets, and new edges matching the  parity-$0$ terminals of
the gadgets (guided by the matching~$\matching$).
The edges connecting the bristles to  parity-$1$ terminals of the gadgets
are  those in the set
$$E_B = \{ ( (u,j), \oddterminal[u,j] ), u \in V, 0 \leq j \leq k_1-1 \}.$$
It is more complicated  to describe the edges matching the  parity-$0$ terminals of the gadgets.
The idea (see Figure~\ref{fig:Jprime})
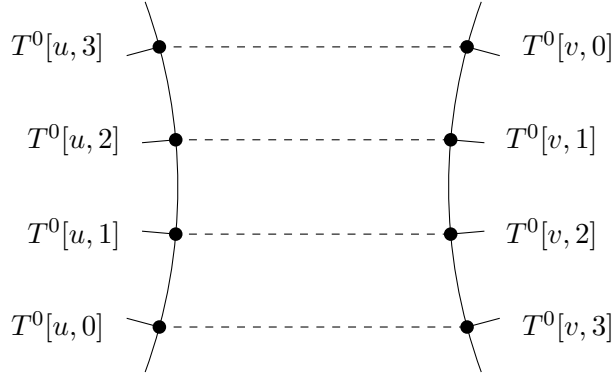
\begin{figure}[!h]
\begin{center}
\begin{tikzpicture}[scale = 0.9]
\useasboundingbox (5,-3) rectangle (15,3);
\draw (canvas polar cs:angle=-20,radius=8cm) arc (-20:20:8cm);
\draw (20,0) + (canvas polar cs:angle=200,radius=8cm) arc (200:160:8cm);
\foreach \t/\othert in {0/3,1/2,2/1,3/0} {
  \pgfmathparse{10*(\t-1.5)} \let\tangle\pgfmathresult
  % draw edge between terminals
  \draw[dashed] (20,0) + (canvas polar cs:angle=180+\tangle,radius=8cm) --
    (canvas polar cs:angle=-\tangle,radius=8cm);
  % lines that hint at internal edges of the gadgets
  \draw (20,0) ++ (canvas polar cs:angle=180+\tangle,radius=8cm) --
    + (canvas polar cs:angle=\tangle,radius=0.5cm);
  \draw (canvas polar cs:angle=-\tangle,radius=8cm) --
    + (canvas polar cs:angle=180-\tangle,radius=0.5cm);
  \draw (canvas polar cs:angle=-\tangle,radius=8cm) + (-1.5,0) node {$T^0[u,\othert$]};
  \draw (21.5,0) + (canvas polar cs:angle=180+\tangle,radius=8cm) node {$T^0[v,\t]$};
  \fill (canvas polar cs:angle=-\tangle,radius=8cm) circle (0.1cm);
  \fill (20,0) + (canvas polar cs:angle=180+\tangle,radius=8cm) circle (0.1cm);
}  

\end{tikzpicture}
\end{center}
\caption{Terminals of $u$ are matched to terminals in $v$, reversing
  the order.}
  \label{fig:Jprime}\end{figure} that if $u_a$ is matched to $v_b$ in $\matching$ 
(where $a\in \{0,1,2\}$ and $b \in \{0,1,2\}$) then
the  parity-$0$ terminals
$\eventerminal[u, a k_2 ],\ldots,\eventerminal[u, a k_2 + k_2-1]$
get matched to the  parity-$0$ terminals
$\eventerminal[u, b k_2],\ldots,\eventerminal[u, b k_2 + k_2-1]$.
However, there is a further complication: To ensure that $J'$ is planar
we must ensure that one of these sequences of terminals is matched in clockwise order, and the
other in anti-clockwise order. 
Thus, let 
$$E_\matching = \{ (\eventerminal[u, a k_2 + j], \eventerminal[v, b k_2 + k_2-1-j])
 \mid u<v, (u_a,v_b) \in \matching, 0\leq j \leq k_2 -1 \}.$$ 
Then $E(J') = E(J) \cup E_B \cup E_\matching$. Note that both~$J$ and~$J'$ are planar as required.

We next show how to define the positive integers~$k_1$, $k_2$ and~$d$.
This is made more complicated because the result in the statement of this Lemma (Lemma~\ref{lem:gadget}) is algorithmic. We have to show how to actually compute all of the numbers used in the reductions including~$k_1$ and~$k_2$ which are parameters of~$J'$. These numbers will be based on the absolute constants $\peq$ and $\pneq$ and on the parameters~$\beta$, $\gamma$ and $\lambda$, all of which we showed can be approximated sufficiently closely in \eqref{goodapprox}. 

Define 
\begin{align*}
P=
\begin{pmatrix}1- \peq& \peq \\1- \pneq & \pneq\end{pmatrix}
\text{, }\qquad
M=
P
\begin{pmatrix}\beta&1\\1&\gamma\end{pmatrix}
P^t
\text{, }\qquad
W = P
\begin{pmatrix}\beta&1\\1&\gamma\end{pmatrix}
\begin{pmatrix}1\\ \lambda\end{pmatrix},
\end{align*}
where $P^t$ denotes the transpose of the matrix $P$.  
Also, define the ``algorithmic versions''
\begin{align*}
\hatP=
\begin{pmatrix} 1 - \hatpeq& \hatpeq \\ 1 - \hatpneq & \hatpneq\end{pmatrix}
\text{, }\qquad
\hatM=
\hatP
\begin{pmatrix}\hatbeta&1\\1&\hatgamma\end{pmatrix}
\hatP^t
\text{, }\qquad
\hatW = \hatP
\begin{pmatrix}\hatbeta&1\\1&\hatgamma\end{pmatrix}
\begin{pmatrix}1\\ \hatlambda\end{pmatrix}.
\end{align*}
Note that if \eqref{goodapprox} holds
then, for any $s\in\{0,1\}$ and $s'\in\{0,1\}$,  \begin{align}
\begin{split}
\label{anothergoodapprox}
e^{-\delta} P_{s,s'} &\leq \hatP_{s,s'} \leq e^\delta P_{s,s'},\\
e^{-\delta} M_{s,s'} &\leq \hatM_{s,s'} \leq e^\delta M_{s,s'},\\
e^{-\delta} W_s &\leq \hatW_s \leq e^\delta W_s.
\end{split}
\end{align}
We will use  \eqref{anothergoodapprox} later to establish an upper bound \eqref{usethis}
on the variation distance between the
distribution that the   (algorithmically-constructed) gadget induces
on configurations on terminals and the idealised distribution. 
 
The matrix~$M$ has the following informal interpretation.
Suppose that two  parity-$0$ terminals $t$ and $t'$ are adjacent in~$J'$.
and that $\sigma\colon \configs{V(J')}$ is a configuration.
If these two terminals have  spins $\sigma(t)$ and $\sigma(t')$, respectively,
then the edge between them contributes a factor
$\left(\begin{smallmatrix}\beta&1\\1&\gamma\end{smallmatrix}\right)_{\sigma(t) ,\sigma(t')}$ 
to 
$w_{J'}(\sigma)$.  
We will show below that, 
if~$t$ is a terminal of~$C[u]$ and the spins of~$C[u]$ are chosen from the idealised
distribution $\mu_u^s$, 
then the probability that the spin of terminal~$t$ is~$j$ is $P_{s,j}$. Thus, informally,
$M_{s,s'}$ captures the expected contribution of this connection (in the
idealised distribution),  where $s'$ represents the  phase of the gadget of terminal~$t'$.

The informal interpretation of $W$ is as follows.
Given any configuration~$\sigma:\configs{V(J)}$, 
each parity-$1$ terminal~$t$ that is connected to a bristle
will contribute a factor 
$\left[(\begin{smallmatrix}\beta&1\\1&\gamma\end{smallmatrix})
(\begin{smallmatrix}1\\ \lambda\end{smallmatrix})\right]_{\sigma(t)}$  
to the sum
$\sum_{\sigma'} w_{J'}(\sigma')$,
where the index of summation ranges over all configurations $\sigma':\configs{V(J')}$
that agree with~$\sigma$ on $V(J)$. This informal description is just
to provide intuition --- the technical details are given below.
The main idea is that, if the spins of the terminals of the gadgets
are chosen from the ``idealised'' distribution then,
if the gadget of~$t$ has  phase~$s$, then 
the terminal~$t$ will contribute a factor of $W_{1\oplus s}$
to the expected contribution from this bristle. 

We now introduce some calculation which will be needed to 
describe the algorithm's computation of~$k_1$, $k_2$, and~$d$
and also to give the definitions of the real numbers 
$\tilde\gamma$ and $\tilde\lambda$.
The first step is deriving some tedious but necessary bounds on the various 
quantities defined above.
In particular, we will define 
positive rational numbers~$\Delta^-$ and~$\Delta^+$ and
a rational number~$\xi\in(0,1)$ 
(independent of~$\delta$ and~$n$, but depending on $\beta$, $\gamma$ and $\lambda$).
These will be hard-wired into the algorithm.
We will  prove that, provided that $\delta$ is sufficiently small, each $\hatM_{s,s'}$ 
and $\hatW_s$  satisfies $\Delta^- \leq \hatM_{s,s'} \leq \Delta^+$ 
and $\Delta^- \leq \hatW_s \leq \Delta^+$.
Also, each of $\hatpeq-\hatpneq$, $1-\hatgamma$, $\hatpneq$ and $\hatlambda$ is
at least~$\xi$. 
We will also prove that $\hatpneq \leq 1$ (and, from~\eqref{goodapprox},
we have $\hatbeta \geq 1$.) 
Finally, we prove $\hatbeta\hatgamma \leq  1-\xi$.
Here are
the details (which the reader may skip).

\medskip
\hrule
\smallskip
start of optional algorithmic technical details
\smallskip
\hrule
\medskip

\begin{itemize}
\item
By Lemmas~\ref{lem:little} and~\ref{lem:peqgreater}, 
we can define positive rational numbers~$p^-$ and~$p^+$ such that
every element $P_{s,s'}$ of matrix~$P$ satisfies 
$p^-\leq P_{s,s'} \leq p^+$.
Then 
$$ {(p^-)}^2(2 + \beta + \gamma)\leq M_{s,s'} \leq {(p^+)}^2(2 + \beta + \gamma),$$
so, since $\delta<1$,
$$ e^{-1}{(p^-)}^2(2 + \beta + \gamma)\leq \hatM_{s,s'} \leq e{(p^+)}^2(2 + \beta + \gamma),$$
so to get the required bounds, we take any $\Delta^- < e^{-1}{(p^-)}^2(2 + \beta + \gamma)$
and any $\Delta^+> e{(p^+)}^2(2 + \beta + \gamma)$. The bounds on $\hatW_s$ are similar.
\item To ensure $\hatpeq - \hatpneq \geq \xi$,
choose rational numbers~$\rho_1$, $\rho_2$, $\rho_3$, and $\rho_4$
such that $\pneq < \rho_1 < \rho_2 < \rho_3 < \rho_4 < \peq$.
These exist by Lemma~\ref{lem:peqgreater}, which guarantees that $\pneq < \peq$.
Then, if $\delta \leq \rho_4 - \rho_3$,
Equation~\eqref{goodapprox} guarantees
$\hatpeq \geq \peq e^{-\delta} \geq \peq(1-\delta) \geq \peq - \delta \geq \peq - (\rho_4 - \rho_3)  \geq \rho_3$.
(Note that the calculation used $\peq \leq 1$.)
Similarly, if $\delta \leq (\rho_2 - \rho_1)/2$,
Equation~\eqref{goodapprox} guarantees
$\hatpneq \leq  e^{\delta} \pneq \leq \pneq(1+2\delta) \leq \pneq + 2\delta \leq \pneq + (\rho_2 - \rho_1) \leq \rho_2$. (Again, we used $\pneq \leq 1$.)
It suffices to take any $\xi \leq \rho_3-\rho_2$.
\item We can similarly  establish $1-\hatgamma\geq \xi$ 
by considering a sequence of rational numbers between $\gamma$ and $1$
(using the fact that $\gamma<1$ by \eqref{eq:spincond})
and we can establish $\hatpneq \geq \xi$ by considering a sequence of rational numbers between~$0$
and~$\pneq$ (using the fact that $\pneq>0$ by Lemma~\ref{lem:little}).
\item Then, by \eqref{eq:spincond}, $\lambda>0$, so 
taking any $\xi < \lambda$,
we can choose a rational number $\xi'$  with $0 <  \xi <  \xi' < \lambda$.
Then choosing $\delta \leq \log(\xi'/ \xi)$ ensures
$e^{-\delta} \xi' \geq \xi$ 
so $\hatlambda \geq e^{-\delta} \lambda \geq e^{-\delta} \xi' \geq \xi$. 
\item Note that the second bullet point already establishes $\hatpneq \leq 1$.
\item Finally, \eqref{eq:spincond} guarantees $\beta\gamma<1$, so
choose rationals $\beta'\geq\beta$ and $\gamma'\geq\gamma$ with $\beta' \gamma'<1$.
Choose~$\xi$ sufficiently small that $\beta' \gamma' \leq e^{-3\xi}$.
Then choose $\delta\leq \xi/2$ to ensure
$$\hatbeta\hatgamma \leq e^{2 \delta} \beta \gamma 
\leq e^{\xi} \beta' \gamma' \leq e^{-2\xi} \leq 1-\xi.$$
 \end{itemize}
 
 \medskip
 \hrule
 \smallskip
 end of optional algorithmic technical details 
 \smallskip
 \hrule
 \medskip
    
We can make the following conclusions.
\begin{align*}
\hatM_{1,1}-\hatM_{0,1}&=( \hatpeq- \hatpneq)((\hatbeta-1)(1- \hatpneq)+(1-\hatgamma) \hatpneq) \geq \xi^3 \\
\hatM_{0,1}-\hatM_{0,0}&=(\hatpeq- \hatpneq)((\hatbeta-1)(1- \hatpeq)+(1-\hatgamma) \hatpeq)  \geq \xi^3\\
\hatW_1-\hatW_0&=( \hatpeq- \hatpneq)((\hatbeta - 1)+(1-\hatgamma)\hatlambda) \geq \xi^3.
\end{align*}

We  can now define~$k_2$.
Since 
$$\hatM_{0,0}\hatM_{1,1}-{\hatM_{0,1}}^2 =\det(\hatM) = 
\det(\hatP)^2 (\hatbeta \hatgamma -1) 
= {(\hatpeq - \hatpneq)}^2 (\hatbeta \hatgamma - 1)
\leq -\xi^3, \mbox{we have}$$

$$
\frac{\hatM_{0,0}\hatM_{1,1} }{\hatM_{0,1}^2}
= \frac{ \hatM_{0,1}^2 + (\hatM_{0,0}\hatM_{1,1}- \hatM_{0,1}^2) }{\hatM_{0,1}^2}
\leq \frac{ \hatM_{0,1}^2 - \xi^3 }{\hatM_{0,1}^2}
= 1- \frac{\xi^3}{\hatM_{0,1}^2}
\leq 1 - \frac{\xi^3}{(\Delta^+)^2} \leq e^{-\xi^3/(\Delta^+)^2}.$$
Then let
$$k_2 = \left\lceil 
\frac{(n^2+n)2 \log(5) {(\Delta^+)}^2}{\xi^3}
\right\rceil.$$
Then,
if we ensure that $\delta <  (\xi^3/(\Delta^+)^2)/8$, 
we have 
\begin{align*} \left( e^{4 \delta}\frac{\hatM_{0,0}\hatM_{1,1}}{ \hatM_{0,1}^2}\right)^{k_2}
& \leq e^{-k_2 \xi^3/(2 (\Delta^+)^2 )}
\\
&\leq   5^{-n^2-n}.
\end{align*}

Then define
$$\tilde \gamma = \left(\frac{M_{0,0}M_{1,1}}{ M_{0,1}^2}\right)^{k_2}.$$
By~\eqref{anothergoodapprox},
$$\tilde\gamma \leq e^{4 \delta k_2} 
  \left(\frac{\hatM_{0,0}\hatM_{1,1}}{ \hatM_{0,1}^2}\right)^{k_2} \leq 5^{-n^2-n}.$$
Also,  there is a randomised approximation scheme for~$\tilde \gamma$
whose running time is at most a polynomial in~$n$ and in the desired accuracy parameter~$\epsilon$. 
   
Next, we will define~$k_1$.
 Recall that $\hatW_1 > \hatW_0$ and $\hatM_{1,1} > \hatM_{0,1}$.
 If $n$ is sufficiently large with respect to
 $\log(\Delta^+/\Delta^-) \geq \log(\hatW_1/\hatW_0)$
 then there is a positive integer~$k_1$ (which the algorithm can compute) satisfying
$$
\frac{3 k_2 \log(\hatM_{1,1}/\hatM_{0,1})}{\log(\hatW_1/\hatW_0)} 
+ \frac{\log(4.1) n}{\log(\hatW_1/\hatW_0)}
\leq k_1 \leq 
\frac{3 k_2 \log(\hatM_{1,1}/\hatM_{0,1})}{\log(\hatW_1/\hatW_0)} 
+ \frac{\log(4.9) n}{\log(\hatW_1/\hatW_0)}.$$ 
Note that $k_1 = O(n^2)$. Also,
 $$
  (4.1)^n \leq \left( \frac{\hatW_1}{\hatW_0}\right)^{k_1}\left( \frac{\hatM_{0,1}}{\hatM_{1,1}}\right)^{3k_2}$$ 
and
$$\left( \frac{\hatW_1}{\hatW_0}\right)^{k_1}\left( \frac{\hatM_{0,1}}{\hatM_{1,1}}\right)^{3k_2} \leq (4.9)^n.$$

Now if we ensure 
 $\delta \leq n^{-2.5}$
then, for sufficiently large~$n$,
$\delta \leq n \log(4.1/4) /(2k_1 + 6 k_2)$
and $\delta \leq n \log(5/4.9)/(2k_1+6k_2)$, so
$$
4^n \leq \left(e^{-2\delta}\frac{\hatW_1}{\hatW_0}\right)^{k_1}\left(e^{-2\delta}\frac{\hatM_{0,1}}{\hatM_{1,1}}\right)^{3k_2}$$ 
and
$$\left(e^{2\delta}\frac{\hatW_1}{\hatW_0}\right)^{k_1}\left(e^{2\delta}\frac{\hatM_{0,1}}{\hatM_{1,1}}\right)^{3k_2} \leq 5^n.$$
Note that 
$$  5^n \leq  \frac{1}{\tilde\gamma^{1/n}}.$$
Then define 
$$ 
 \tilde\lambda=\left(\frac{W_1}{W_0}\right)^{k_1}\left(\frac{M_{0,1}}{M_{1,1}}\right)^{3k_2}.$$ 
Note that there is a randomised approximation scheme for $\tilde\lambda$
whose running time is bounded from above by a polynomial in~$n$ and
the desired accuracy parameter~$\epsilon$. Also,  $\tilde \lambda \leq \frac{1}{\tilde\gamma^{1/n}}$ so 
$\tilde\gamma\leq {\tilde\lambda}^{-n}$
and $\tilde \lambda \geq 4^n$,
as required.

Now let $k= \max(k_1,3 k_2)$.
Finally, the gadget will use a parameter~$d$.
By Proposition~\ref{prop:gadget}, 
there is a $c>1$ (not depending on $k$) such that, for all sufficiently large $d$ which are
multiples of~$16$, and all configurations $\tau:\configs{\terminals}$,
$$|\Pr(\sigmarv(\terminals)=\tau)-\muall(\tau)| \leq c^{-d} k^2.$$
The algorithm will choose 
$d$ to be a multiple of~$16$ such that $d = O(n^3)$  and
$$ c^{-d} k^2 < \frac{{(e^{-\delta} \hatW_0)}^{k_1 n} {{(e^{-\delta} \hatM_{1,1})}}^{k_2 |E|}}
{({e^{\delta}(\hatbeta + \hatlambda))}^{k_1 n} {(e^{\delta}\hatbeta)}^{k_2 |E|} 2^{2 k n} 2^n n}   .$$
This can be done, since $|E| = O(n)$.
We will use below the fact that
 \begin{equation}
\label{usethis}
\max_{\tau:\configs{\terminals}} 
|\Pr(\sigmarv(\terminals)=\tau)-\muall(\tau)|  <
\frac{W_0^{k_1 n} M_{1,1}^{k_2 |E|}}
{{(\beta + \lambda)}^{k_1 n} \beta^{k_2 |E|} 2^{2 k n}2^n n},
\end{equation}
which follows from Equations~\eqref{goodapprox} and~\eqref{anothergoodapprox}.
  
Let $K$ be the  positive real
given by 
$$K = \frac{2^n}{W_0^{k_1n}M_{1,1}^{k_2|E|}}.$$
Note that there is a randomised approximation scheme for~$K$ whose running time
is at most a polynomial in~$n$ and in the desired accuracy parameter, $\epsilon$.

All that remains is to  establish~\eqref{eq:lem15}, which we do in the remainder of the proof.
Let $T =  \cup_{u\in V} T[u]  $ be the set of terminals in $V(J)$.
For every configuration $\tau: \configs{T}$, let   
$\wt(\tau) = \sum_{\sigma\in  \configs{V(J)}: \sigma(T)=\tau} w_J(\sigma)$.
The quantity~$\wt(\tau)$ is the contribution to $Z_{\beta,\gamma,\lambda}(J)$ from
configurations $\sigma$ with $\sigma(T)=\tau$.
Similarly, let
$\wt'(\tau)= \sum_{\sigma'\in \configs{V(J')}: \sigma'(T)=\tau} w_{J'}(\sigma')$ 
be the contribution to $Z_{\beta,\gamma,\lambda}(J')$ from
the corresponding configurations on~$J'$.  
Let $F(\tau)$ denote
${\wt'(\tau)}/{\wt(\tau)}$.  Then
\begin{align}
\frac
{Z_{\beta,\gamma,\lambda}(J')}
{Z_{\beta,\gamma,\lambda}(J)}
=
\frac{\sum_{\tau:\configs{T}}\wt'(\tau)}
{ Z_{\beta,\gamma,\lambda}(J)}
=
\frac{\sum_{\tau:\configs{T}}\wt(\tau) F(\tau)}
{ Z_{\beta,\gamma,\lambda}(J)}
=
\ex[F(\bsigma_{J}(T))]\label{eq:F}.
\end{align}

We can write $F(\tau)$ in terms of $\beta$, $\gamma$, and~$\lambda$:
\begin{align}
F(\tau)&=\left[\prod_{u\in V}\prod_{j=0}^{k_1-1}
\left(
\begin{pmatrix}\beta&1\\1&\gamma\end{pmatrix}
\begin{pmatrix}1\\ \lambda\end{pmatrix}
\right)_{\tau( \oddterminal[u, j])}\right]\notag\\
&\qquad\qquad\times\left[\prod_{(u_a,v_b)\in \matching,u<v}\prod_{j=0}^{k_2-1}
\begin{pmatrix}\beta&1\\1&\gamma\end{pmatrix}_{
\tau(\eventerminal[u, a k_2+j]), \tau(\eventerminal[v, b k_2 + k_2-1-j])}\right].
\label{eq:F_terminals}
\end{align}
 
We now define an ``idealised'' distribution on configurations assigning spins to the
terminals.
First, let $\bsigmatilde$
be a random variable which is drawn uniformly from $\configs{V}$.
Each realisation $\sigmatilde$ of $\bsigmatilde$ can be thought of 
as specifying, for every vertex $u\in V$, a phase $\sigmatilde(u)\in\{0,1\}$ 
for the gadget~$C[u]$.  
Conditioned on the realisation $\bsigmatilde=\sigmatilde$,
the random variable $\bsigmahat:\configs{T}$ is distributed as follows:
for each $u\in V$,
$\bsigmahat(T[u])$ is chosen independently from the distribution~$\mu_u^{\sigmatilde(u)}$.
Note that the (unconditioned) random variable~$\bsigmahat$ has the property 
that $\bsigmahat(T[u])$ is distributed as~$\mu_u$, independently of all $\bsigmahat(T[u'])$ 
for $u'\not=u$.

We wish to estimate
$\ex[F(\bsigma_{J}(T))]$, but the distribution of~$\bsigma_J(T)$ is somewhat complicated.
Instead, we will  first estimate
$\ex[F(\bsigmahat)]$, and we will later use Proposition~\ref{prop:gadget} to show that
these two quantities are close.  From the
definition of~$\bsigmahat$, we have
$$\ex[F(\bsigmahat)] = \frac{1}{2^n}\sum_{\sigmatilde:\configs{V}} \ex[F(\bsigmahat)\mid \bsigmatilde=\sigmatilde].$$
Now we will argue that
if $\sigmatilde(u)=s$ then, for every  parity-$0$ terminal~$t$ of~$C[u]$, it is the case that
$\Pr(\bsigmahat(t)=j\mid \bsigmatilde=\sigmatilde) = P_{s, j}$.
(To see this, consider the possible cases. If $s=0$ then,
from the definition of $\muall^0$, the probability that $\bsigmahat(t)=1$
is $\peq$, which is $P_{0,1}$, but the probability that $\bsigmahat(t)=0$ is
$1-\peq = P_{0,0}$. The situation is similar if $s=1$.)
On the other hand, similar reasoning shows that, for every  parity-$1$ terminal~$t$ of~$C[u]$,
it is the case that
$\Pr(\bsigmahat(t)=j\mid \bsigmatilde=\sigmatilde) = P_{1\oplus s,j}$.
Thus,   we have 
$$ \ex[F(\bsigmahat)] = \frac{1}{2^n}\sum_{\sigmatilde: \configs{V}}
\sum_{\tau:\configs{T}} F(\tau)
\prod_{u\in V}
\prod_{j=1}^{k} P_{\sigmatilde(u),\tau(\eventerminal[u,j])}
   P_{1\oplus \sigmatilde(u),\tau(\oddterminal[u,j])} .$$

Plugging in \eqref{eq:F_terminals},
the contribution  of each $\tilde\sigma$ to the right-hand-side of the above equality is
$2^{-n}$ multiplied by the product of the following terms: 
\begin{align*}
&\prod_{u \in V} \prod_{j=0}^{k_1-1} \sum_{s\in\{0,1\}} P_{1\oplus\tilde\sigma(u),s}
\left(
\begin{pmatrix}\beta&1\\1&\gamma\end{pmatrix}
\begin{pmatrix}1\\ \lambda\end{pmatrix}
\right)_{s},\\
&\prod_{u\in V} \prod_{j=k_1}^{k-1} \sum_{s\in\{0,1\}} P_{1\oplus\tilde\sigma(u),s},\\
&\prod_{(u,v)\in E} \prod_{j=0}^{k_2-1}
\sum_{s\in\{0,1\}} \sum_{s'\in\{0,1\}} P_{\tilde\sigma(u),s} P_{\tilde\sigma(u),s'} 
\begin{pmatrix}\beta&1\\1&\gamma\end{pmatrix}_{s,s'},\\
&\prod_{u\in V} \prod_{j= 3 k_2}^{k-1} \sum_{s\in\{0,1\}} P_{\tilde\sigma(u),s}.\\
\end{align*}
The second and fourth of these terms are equal to~$1$, and the first and third simplify using the
matrices that we defined earlier, so we get
$$
\ex[F(\bsigmahat)] = \frac{1}{2^n}\sum_{\tilde\sigma: \configs{V}}
\prod_{u \in V} W_{1\oplus\tilde\sigma(u)}^{k_1}
\prod_{(u,v)\in E} {M_{\tilde\sigma(u),\tilde\sigma(v)}}^{k_2}.$$

Then, plugging in our notation from earlier, we have
\begin{align*}
\ex[F(\bsigmahat)]
&=
\frac{1}{2^n}
\sum_{\tilde\sigma:\configs{V}}
W_1^{k_1(n-\ell(\tilde\sigma))}
W_0^{k_1\ell(\tilde\sigma)}
M_{0,0}^{k_2b(\tilde\sigma)}
M_{0,1}^{k_2(|E|-b(\tilde\sigma)-c(\tilde\sigma))}
M_{1,1}^{k_2c(\tilde\sigma)}\\
&=
\frac{W_0^{k_1n}M_{0,1}^{k_2|E|}}{2^n}
\sum_{\tilde\sigma:\configs{V}}
\left(\frac{M_{0,0}}{M_{0,1}}\right)^{k_2b(\tilde\sigma)}
\left(\frac{M_{1,1}}{M_{0,1}}\right)^{k_2c(\tilde\sigma)}
\left(\frac{W_1}{W_0}\right)^{k_1(n-\ell(\tilde\sigma))}.
\end{align*}
Replacing $\tilde\sigma(u)$ with $1\oplus\tilde\sigma(u)$, we get
$$\ex[F(\bsigmahat)] =
\frac{W_0^{k_1n}M_{0,1}^{k_2|E|}}{2^n}
\sum_{\tilde\sigma:\configs{V}}
{\left(\frac{M_{1,1}}{M_{0,1}}\right)^{k_2b(\tilde\sigma)}}
{\left(\frac{M_{0,0}}{M_{0,1}}\right)^{k_2c(\tilde\sigma)}}
{\left(\frac{W_1}{W_0}\right)^{k_1\ell(\tilde\sigma)}}.$$

Since $G$ is cubic, we can count the pairs $(v,(u,v))\in V\times E$
with $\tilde\sigma(v)=1$ in two ways to get
$3\ell(\tilde\sigma)=2c(\tilde\sigma)+(|E|-b(\tilde\sigma)-c(\tilde\sigma))$.
So $b(\tilde\sigma)=c(\tilde\sigma)+|E|-3\ell(\tilde\sigma)$,
which implies that 
\begin{align}\nonumber
\ex[F(\bsigmahat)] &=
  \frac{W_0^{k_1n}M_{0,1}^{k_2|E|}}{2^n}
\left(\frac{M_{1,1}}{M_{0,1}}\right)^{k_2|E|} 
 \sum_{\sigma:
\configs{V}}
\left[\left(\frac{M_{0,0}M_{1,1}}{M_{0,1}^2}\right)^{k_2}\right]^{c(\sigma)}
\left[\left(\frac{W_1}{W_0}\right)^{k_1}\left(\frac{M_{0,1}}{M_{1,1}}\right)^{3k_2}\right]^{\ell(\sigma)} 
\\ &=K^{-1}Z_{(1,\tilde\gamma,\tilde\lambda)}(G).\label{eq:F_ideal}
\end{align}

Plugging~\eqref{eq:F} and~\eqref{eq:F_ideal} into~\eqref{eq:lem15},
it remains to prove
\begin{equation}
\label{eq:whatsleft}
e^{-1/4}  \ex[F(\bsigmahat)] \leq \ex[F(\bsigma_{J}(T))]
\leq e^{1/4} \ex[F(\bsigmahat)].
\end{equation}
Let
$\psi=|\ex[F(\bsigma_{J}(T))]-\ex[F(\bsigmahat)]|$. 
Now
 $$
\psi\leq \left(\max_{\tau:\configs{T} }F(\tau)\right) \sum_{\tau:\configs{T}}
\left| \Pr(\bsigma_{J}(T)=\tau)-\Pr(\bsigmahat=\tau)\right|.
$$
To emphasise that summation over $\tau:\configs{T}$ can be
broken into summation over each restriction $\tau(T[u])$,
we will write the summation index as $\forall u, \tau(T[u]):\configs{T[u]}$.
By~\eqref{eq:F_terminals}, 
$F(\tau) \leq 
{(\beta+\lambda)}^{k_1 n}
\beta^{k_2 |E|}$, so we can write
\begin{align*}
\psi &\leq  {(\beta+\lambda)}^{k_1 n}
\beta^{k_2 |E|}
 \sum_{ \forall u, \tau(T[u]):\configs{T[u]}}
\left| \prod_{u\in V}
\Pr(\bsigma_{J}(T[u])= \tau(T[u]))-
\prod_{u\in V}
\Pr(\bsigmahat(T[u])=\tau(T[u]))\right|\\
&= {(\beta+\lambda)}^{k_1 n}
\beta^{k_2 |E|}
 \sum_{\forall u, \tau(T[u]):\configs{T[u]}
   }
\left| 
\prod_{u\in V}
\Pr(\bsigma_{J}(T[u])= \tau(T[u]))
  - 
\prod_{u \in V}
\mu(\tau(T[u])) \right|.
\end{align*}

 Using the inequality 
 $| \prod a_u-\prod b_u |\leq \sum | a_u - b_u| $ valid for values $0\leq a_u,b_u\leq 1$, as in the proof of Proposition~\ref{prop:gadget}, we have
$$
\psi \leq 
{(\beta+\lambda)}^{k_1 n}
\beta^{k_2 |E|}
 2^{2 k n } n  
 \max_{\tau:\configs{T_{k,d}}}
 \left|
 \Pr(\bsigmakd=\tau)
 - \muall(\tau) \right|.$$ 
Applying  \eqref{usethis},
$$\psi \leq
{(\beta+\lambda)}^{k_1 n}
\beta^{k_2 |E|}
 2^{2 k n } n 
 \left(
 \frac{W_0^{k_1 n} M_{1,1}^{k_2 |E|}}
{{(\beta + \lambda)}^{k_1 n} \beta^{k_2 |E|} 2^{2 k n}2^n n} 
 \right)  =
 \frac{W_0^{k_1n}M_{1,1}^{k_2|E|}}{2^n} = K^{-1}.$$
So to establish~\eqref{eq:whatsleft},
we note
that
$$
1-\frac{1}{Z_{1,\tilde\gamma,\tilde\lambda}(G)}
=
1-\frac{1}{K \ex[F(\bsigmahat)]} \leq
\frac{\ex[F(\bsigma_{J}(T))]}{\ex[F(\bsigmahat)]}
\leq 1 + \frac{1}{K \ex[F(\bsigmahat)]}
= 1 + \frac{1}{Z_{1,\tilde\gamma,\tilde\lambda}(G)}.$$
The result follows from the extremely crude bound $Z_{1,\tilde\gamma,\tilde\lambda}(G)\geq 8$.
\end{proof} 

We can now prove our main theorem.

\begin{mainthm}
Let $\beta$, $\gamma$ and $\lambda$ be efficiently-approximable reals
satisfying \eqref{eq:spincond}.
There is no FPRAS for $\PlanarTwoSpin(\beta,\gamma,\lambda)$
unless $\NP=\RP$.
\end{mainthm}
 
\begin{proof}
We will give a randomised algorithm for $\PlanarCubicIS$, using
an FPRAS for $\PlanarTwoSpin(\beta,\gamma,\lambda)$ as an oracle
(and also using the given FPRASes for $\beta$, $\gamma$ and $\lambda$).

After receiving an instance~$G$ and~$h$,
our algorithm uses Lemma~\ref{lem:gadget}  
which provides planar graphs $J$ and~$J'$ with maximum degree at most~$4$
and also some approximation schemes for the reals $K$, $\tilde \gamma$ 
and $\tilde \lambda$.
With probability at least~$1-1/15$,
these satisfy  
$\tilde\lambda \geq 4^{|V(G)|}$ and $\tilde\gamma\leq {\tilde\lambda}^{-|V(G)|} < 1$
and  
Equation~\eqref{eq:lem15}. 
The algorithm then makes four calls to 
approximation schemes, 
suitably powered so that each call fails with probability at most~$1/15$.
Thus, with probability at least $2/3$, 
Equation~\eqref{eq:lem15} is satisfied and 
all calls to the approximation schemes succeed.
In that case, we will show how to determine
(from the outputs of the 
approximation schemes) whether or not $G$ has an independent set of size~$h$.
 
Let $n=|V(G)|$.   
By Equation~\eqref{eq:lem15}, we have
\begin{equation*}
 e^{-1/4} 
K \frac{Z_{\beta,\gamma,\lambda}(J')}{Z_{\beta,\gamma,\lambda}(J)}
\leq  Z_{1,\tilde\gamma,\tilde\lambda}(G)
\leq e^{1/4}
K \frac{Z_{\beta,\gamma,\lambda}(J')}{Z_{\beta,\gamma,\lambda}(J)}.
 \end{equation*}

Using the given approximation schemes for $Z_{\beta,\gamma,\lambda}(J')$,
$Z_{\beta,\gamma,\lambda}(J)$ and $K$,
each with accuracy parameter~$1/12$ and 
failure probability~$1/15$, 
we can compute a value $\hat{Z}$ which, 
with probability at least $1-3/15$,
satisfies
$$e^{-1/4} \hat{Z} \leq  K \frac{Z_{\beta,\gamma,\lambda}(J')}{Z_{\beta,\gamma,\lambda}(J)}
\leq e^{1/4} \hat{Z},$$
so
\begin{equation}
\label{blahone}
e^{-1/2} \hat{Z} \leq  Z_{1,\tilde\gamma,\tilde\lambda}(G)
\leq e^{1/2} \hat{Z}.
\end{equation}

Using the given  approximation scheme for $\tilde\lambda$ with accuracy parameter~$1/2h$ and failure probability~$1/15$,
we can
compute a value $\hat\lambda$
which, with probability at least $1-1/15$, satisfies
$e^{- 1/2h} \hat\lambda \leq \tilde \lambda \leq e^{ 1/2h} \hat\lambda$
so 
\begin{equation}
\label{blahtwo}
e^{- 1/2} {\hat\lambda}^h \leq {\tilde \lambda}^h \leq e^{ 1/2} {\hat\lambda}^h.
\end{equation}

Suppose that all four calls to the  approximation schemes succeed so that \eqref{blahone} and \eqref{blahtwo} hold. 
Recall that 
$$ Z_{1,\tilde\gamma,\tilde\lambda}(G)
= \sum_{\sigma:  \configs{V(G)} } 
  {\tilde\gamma}^{c(\sigma)} {\tilde\lambda}^{\ell(\sigma)}.$$

If $G$ has an independent set of size~$h$ then
$Z_{1,\tilde\gamma,\tilde\lambda}(G)
\geq   {\tilde\lambda}^h$ so, plugging in \eqref{blahone} and \eqref{blahtwo},
 $$
 \hat Z \geq e^{- 1/2}
 Z_{1,\tilde\gamma,\tilde\lambda}(G)
\geq e^{- 1/2} {\tilde\lambda}^h
\geq e^{- 1} {\hat\lambda}^h
.$$
Also, if $G$ has no independent set of size~$h$, then
$Z_{1,\tilde\gamma,\tilde\lambda}(G)
 \leq
2^n \max( {\tilde\lambda}^{h-1}, {\tilde\lambda}^{n} \tilde \gamma )
\leq 2^n  {\tilde\lambda}^{h-1}
 $. So, plugging in \eqref{blahone} and \eqref{blahtwo} and our lower bound for $\tilde \lambda$,
 $$ \hat Z \leq e^{ 1/2}  Z_{1,\tilde\gamma,\tilde\lambda}(G) \leq
e^{1/2} 2^n {\tilde\lambda}^{h-1} \leq 
e^{1/2} \frac{2^n {\tilde \lambda}^h}{4^n} 
 = e^{1/2} 2^{-n} {\tilde \lambda}^h 
 \leq e 2^{-n} {\hat \lambda}^h.$$
 
 As long as $n\geq 3$,  
 $e^{-1} > e 2^{-n}$, so
 it is possible to determine from $\hat Z$ and $\hat \lambda$ whether
 or not $G$ has an independent set of size $h$.

The reduction described above provides a randomised algorithm for 
$\PlanarCubicIS$ with 2-sided error. In the event that 
$\PlanarTwoSpin(\beta,\gamma,\lambda)$ has an FPRAS, the reduction would
place $\PlanarCubicIS$ in BPP\null.
However, the inclusion $\NP\subseteq \BPP$  would imply $\NP=\RP$ \cite[Theorem 2]{Ko}.
So, for any fixed $\beta$, $\gamma$ and $\lambda$ satisfying \eqref{eq:spincond},
there is no FPRAS for $\PlanarTwoSpin(\beta,\gamma,\lambda)$   unless
$\NP=\RP$.
\end{proof}

\section{Approximating the  log-partition function}\label{sec:twospin_log_fpras}
 
We start with a preliminary lemma, which will help us to show that
our approximation is sufficiently accurate.
 
\begin{lemma}\label{lem:logbound}
Suppose that $\beta$, $\gamma$ and $\lambda$ are
real numbers satisfying $\beta\geq 1> \gamma \geq 0$ and $\lambda\geq 1$.
Then, for every planar graph~$G$,
$Z_{\beta,\gamma,\lambda}(G)  \geq
(1+\lambda)^{|V(G)|/4}$.
\end{lemma}
\begin{proof}
Let~$I$ be the largest colour class in a proper 4-colouring of $G$.
Then~$I$ is an independent set of~$G$ of size at least~$|V(G)|/4$.
For every configuration $\sigma\colon\configs{V(G)}$  
which assigns spin~$0$ to every vertex in $V(G)\setminus I$,
$w_{G}(\sigma) \geq \lambda^{\ell(\sigma)}$.  
Thus
$Z_{\beta,\gamma,\lambda}(G)\geq (1+\lambda)^{|I|}$.
\end{proof}

Our approximation algorithm is inspired by Baker's approximation schemes for
optimisation problems on planar graphs~\cite{Baker}. For a good explanation of her
technique (which we use in our exposition here), see Borradailes's notes~\cite{BakerNotes}.  
We will use the following notation (from~\cite{BakerNotes}) 
to decompose a planar graph~$G=(V,E)$ 
which is embedded in the plane.
We first define the \emph{level} of each vertex. Vertices on the boundary of the embedding
have level~$0$.
Then, for $i\in \{0,\ldots,n\}$,
the vertices with level~$i$ are those that are on the boundary
on the graph formed from $G$ by deleting all vertices whose level is less than~$i$.
For a fixed parameter~$k$, and for every $i\in \{0,\ldots,k-1\}$,
let $V_i = \{ v\in V \mid 
\mbox{The level of vertex $v$ is equal to $i$ modulo~$k$}\}$.
Let ${G}_i$ be the graph $G-V_i$.
By construction, ${G}_i$ is $(k-1)$-\emph{outerplanar}.
Also, 
Bodlaender~\cite{Bodlaender} had shown
that every  $k$-outerplanar  graph has treewidth at
most $3k-1$.  Also, this tree decomposition is easy to compute.
Using a data structure of Lipton and Tarjan~\cite{LT},
Baker shows that the levels of vertices can be computed in $O(|V|)$ time.

We can now prove Theorem~\ref{thm:twospin_log_pras}.
  
\begin{thmthree}
Suppose that $\beta$, $\gamma$ and $\lambda$ are efficiently
approximable reals satisfying  
$\beta\geq 1> \gamma\geq 0$ and $\lambda\geq 1$.
There is a PRAS for $\PlanarLogTwoSpin(\beta,\gamma,\lambda)$.
\end{thmthree}

\begin{proof}

Consider input $G=(V,E)$ with at least~$3$ vertices and an accuracy parameter $\epsilon\in (0,1)$.
Let $n=|V|$ 
and 
$m=|E| \leq 3n$.
Let $\beta^+$, $\beta^-$, $\lambda^+$, $\lambda^-$, 
$\gamma^+$ and $\gamma^-$
be rational numbers (built into the algorithm) such that
$\beta^+ \geq \beta \geq \beta^-\geq 1$,
$1>\gamma^+ \geq \gamma \geq \gamma^-\geq 0$, and
$\lambda^+ \geq \lambda \geq \lambda^- \geq 1$.
Let $k$ be any integer satisfying
$$k \geq \frac
{32 \log (2 \lambda^+) + 96 \log(\beta^+)}
{\epsilon \log(1+\lambda^-)}.$$ 
Then let  
$$\delta = \frac{2 n \log(2 \lambda^+)}{k(n+m)}.$$

Using the given FPRASes for $\beta$, $\gamma$ and $\lambda$,
compute $\hatbeta$, $\hatgamma$ and $\hatlambda$
satisfying
$e^{-\delta} \beta \leq \hatbeta \leq e^{\delta} \beta$,
$e^{-\delta} \lambda \leq \hatlambda \leq e^\delta \lambda$ and
$e^{-\delta} \gamma \leq \hatgamma \leq e^\delta \gamma$.
As in the proof of  Lemma~\ref{lem:gadget}, 
adjust the output of the FPRASes to ensure 
$\beta^+ \geq \hatbeta \geq \beta^-$,
$\gamma^+ \geq \hatgamma \geq \gamma^-$, and
$\lambda^+ \geq \hatlambda \geq \lambda^-$.

The first step is to compute a value $\hatZ$
satisfying \begin{align}
\label{eq:bakerbound} 
\hatZ \leq Z_{\hatbeta,\hatgamma,\hatlambda}(G) \leq  {(2 \lambda^+)}^{2n/k} {(\beta^+)}^{12n/k}   \hatZ.
\end{align}
This step is accomplished  as follows.
\begin{enumerate}
\item Using Baker's algorithm, construct the graphs
${G}_i$  for $i\in \{0,\ldots,k-1\}$. Each of these has treewidth at most $3(k-1)-1$.
\item Choose $i\in \{0,\ldots,k-1\}$ as follows.
Let $\mathcal{I} = \{ i \mid\, |V_i| \leq 2n/k \}$.
Note that $|\mathcal{I}| \geq k/2$.
Now consider the $2m$ endpoints of edges in $E$.
Choose $i\in \mathcal{I}$
so that $V_i$ contains at most $(2m)/|\mathcal{I}|$ of these.
Note that $|V_i| \leq 2n/k$ and 
the number of edges with endpoints in $V_i$ is at most $4m/k\leq 12n/k$.
\item Use the algorithm of Yin and Zhang~\cite[Theorem 2.8]{YZ12}
to compute $\hatZ=Z_{\hatbeta,\hatgamma,\hatlambda}({G}_i)$.
The running time of Yin and Zhang's algorithm is at most 
the product of a polynomial in~$n$ and an exponential function in the
treewidth of~$G_i$. In order to apply the algorithm, we first express
the partition function $Z_{\hatbeta,\hatgamma,\hatlambda}(G_i)$ as
the solution to a Holant problem Holant$(\mathcal{G},\mathcal{F})$ with regular symmetric~$\mathcal{F}$.
The partition function of any $2$-spin system can 
be expressed as the partition function of a Holant problem --- see \cite{YZ12} for definitions and details.
\item Equation~\eqref{eq:bakerbound} now follows  
by noting that
$$Z_{\hatbeta,\hatgamma,\hatlambda}(G) = \sum_{\tau\colon\configs{V(G_i)}} w_{G_i}(\tau)
\sum_{\tau':\configs{V_i}} \hatlambda^{\ell(\tau')} {\hatbeta}^{b(\tau,\tau')} {\hatgamma}^{c(\tau,\tau')},$$
where $\ell(\tau')$ is the number of vertices $u\in V_i$ with $\tau'(u)=1$ and
$b(\tau,\tau')$ is  
the sum of the number of edges $(u,v)$ with $u\in V(G_i)$ and $v\in V_i$ and
  $\tau(u)=\tau'(v)=0$
and the number of edges $(u,v)$
with $u\in V_i$ and $v\in V_i$ and  $\tau'(u)=\tau'(v)=0$
and $c(\tau,\tau')$ is defined similarly (with spin~$1$). 
Then $\sum_{\tau':\configs{V_i}} \hatlambda^{\ell(\tau')} {\hatbeta}^{b(\tau,\tau')} {\hatgamma}^{c(\tau,\tau')}$
is at least~$1$ (since $\tau'$ can assign spin~$0$ to every vertex in~$V_i$)
and it is at most $ {2^{2n/k} (\hatlambda)}^{2n/k} {\hatbeta\,}^{12n/k}$.
\end{enumerate}

To finish, note that
$$e^{-\delta(n+m)} Z_{\beta,\gamma,\lambda}(G) \leq Z_{\hatbeta,\hatgamma,\hatlambda}(G) \leq
e^{\delta(n+m)} Z_{\beta,\gamma,\lambda}(G),$$
so
since 
$$\delta (n+m) \leq \frac{2n}{k} \log(2 \lambda^+) $$
and 
(from Lemma~\ref{lem:logbound})
$$\log Z_{\beta,\gamma,\lambda}(G) \geq (n/4) (1+\lambda^-),$$
and
$$(\epsilon/2) (n/4) \log(1+\lambda^-) \geq \frac{4n}{k} \log(2 \lambda^+)
+ \frac{12n}{k} \log(\beta^+),$$ 

\begin{align*}
e^{-\epsilon} \log(Z_{\beta,\gamma,\lambda}(G)) 
&\leq \log(Z_{\beta,\gamma,\lambda}(G))(1-\epsilon/2) \\
&\leq \log Z_{\beta,\gamma,\lambda}(G) - \delta(n+m) 
- \frac{2n}{k} \log(2 \lambda^+)
- \frac{12n}{k} \log(\beta^+)\\
& \leq \log Z_{\hatbeta,\hatgamma,\hatlambda}(G) - \frac{2n}{k} \log(2 \lambda^+)
- \frac{12n}{k} \log(\beta^+)\\
& \leq \log \hatZ.
\end{align*}

Similarly,
\begin{align*}
\log \hatZ &\leq \log Z_{\hatbeta,\hatgamma,\hatlambda}(G)\\
&\leq \delta(n+m) + \log Z_{\beta,\gamma,\lambda}(G)\\
& \leq \frac{2n}{k} \log(2 \lambda^+) + \log Z_{\beta,\gamma,\lambda}(G)\\
& \leq (1+\epsilon/2) \log Z_{\beta,\gamma,\lambda}(G)\\
& \leq e^\epsilon \log Z_{\beta,\gamma,\lambda}(G).
\end{align*}
\end{proof}

\bibliographystyle{plain}
\bibliography{\jobname}
 
\end{document}